\documentclass[11pt,a4paper]{article}

\usepackage[utf8]{inputenc}
\usepackage[T1]{fontenc}
\usepackage{amsmath,amssymb,amsthm}
\usepackage{booktabs}
\usepackage{graphicx}
\usepackage[margin=2.4cm]{geometry}
\usepackage{array}
\usepackage{float}
\usepackage{caption}
\usepackage[colorlinks=true,linkcolor=blue,citecolor=blue,urlcolor=blue]{hyperref}
\usepackage{verbatim}
\usepackage[ruled,linesnumbered]{algorithm2e}

\usepackage{iftex}
\ifxetex
  \usepackage{xeCJK}
\fi

\newtheorem{theorem}{Theorem}
\newtheorem{proposition}{Proposition}
\newtheorem{definition}{Definition}
\newtheorem{hypothesis}{Hypothesis}

\DeclareMathOperator*{\argmax}{arg\,max}

\title{Reversibility-Verified De-identification for Cloud-Local LLM Inference:\\
A Locally Certified Dehydrate-Rehydrate Loop with Layered Assurance (DR-SL)}
\author{Wen Hu$^{1,2}$\quad Ya Yu$^{1}$\quad Xutong Wang$^{1}$\\[2pt]
\small $^{1}$ College of Electronic and Information Engineering,\\
\small Nanjing University of Aeronautics and Astronautics, Nanjing, China\\
\small $^{2}$ Jiangsu Yunhefeng Intelligent Technology Co., Ltd., Jiangsu, China\\
\texttt{\small huwen@nuaa.edu.cn}\quad
\texttt{\small 15949116139@163.com}\quad
\texttt{\small wsgxt@outlook.com}}
\date{}

\begin{document}
\maketitle

\begin{abstract}
Cloud-local LLM inference must keep sensitive user data on-device while exploiting cloud-grade reasoning, yet existing sanitization approaches---placeholder substitution, differential-privacy perturbation, skill distillation---lack a release decision that is simultaneously safe and utility-preserving. We propose \textbf{DR-SL} (Dehydrate-Rehydrate with Self-Learning loop), which formalizes de-identification completeness as two measurable conditions---\textbf{de-identification sufficiency} under Pufferfish semantics and \textbf{task-information preservation} via QA probes---and decides releases with a fully local two-branch verifier under a lexicographic gate with guaranteed termination, backed by a deterministic hard line, an external strong-attacker re-test, and human fallback. Theoretically we prove Fano-type lower bounds, a Pufferfish witness, and a rate-privacy feasibility criterion, and we state their scope plainly: the bounds certify leakage, never safety, and are near-vacuous at our operating point---release safety rests on empirical calibration, the hard line, and human review. On a worst-case fully task-coupled benchmark, the loop reduces leakage $0.457\to0.304$ ($p\approx0$) and the release chain delivers \textbf{0.000} literal leakage at egress (160 instances, two strong attackers), with the system honestly degrading to a certification-and-routing mode exactly as the feasibility criterion predicts. On a mixed-coupling benchmark the same safe point releases \textbf{67.5\%} of instances automatically at zero measured leakage, Pareto-dominating placeholder and selective-LDP corners under an identical release rule. Two human studies anchor the semantic utility metric (Spearman $\rho=0.839$) and the annotation gold (type-level recall ${\ge}0.987$). The exploratory self-learning hypothesis was not supported and is reported as such. All theory bounds pass executable numerical verification; code, synthetic datasets, protocol, and human-study packages are public.

\end{abstract}

\section{Introduction}

\subsection{Background and Motivation}

A fundamental security tension exists between the reasoning capability of large language models and their deployment paradigms. Cloud models (proprietary APIs or large open-weight models) offer stronger reasoning, planning, and tool use, but user requests must leave the local device, exposing sensitive data to interception, retention, and misuse by the service provider; local small models (SLMs) protect data on-device but reason weakly. Cloud-local LLM inference systems~\cite{hao2024hybrid,jin2025cecollm,li2025collab,siyan2025papillon} seek both cloud-grade reasoning and on-device data protection, under the core constraint that \textbf{cloud-bound requests must exclude personally identifiable information (PII) to prevent external data leakage}~\cite{das2025security,lukas2023pii}. Yet the question ``how to guarantee the cloud receives only non-sensitive content'' still lacks a decision mechanism that is both safe and practical: a lax decision leaks; a strict one strips task semantics and degrades reasoning quality.

The problem is acute in real scenarios such as medical consultation and financial services. For example, a user types on a local device: ``I am 67, had a stent procedure in the cardiology department of XX Hospital three months ago, and take aspirin daily; I would like advice on post-operative exercise.'' The combination of age, hospital, surgical history, and medication uniquely identifies the individual even without a name or ID number. Forwarding it to the cloud constitutes quasi-identifier-level leakage~\cite{sweeney2002k}; crudely deleting all details leaves the cloud unable to give any useful advice.

\subsection{Three Fundamental Defects of Existing Approaches}

Existing methods fall into three classes. First, entity-level placeholder substitution~\cite{presidio,rehydra,zeng2025privacyrestore} (e.g., Rehydra, Presidio) replaces names, emails, and keys with placeholders and restores them later. It is deterministic and auditable, but covers only \emph{declared} identifiers; it cannot handle semantic/quasi-identifier leakage---where the dehydrated text contains no explicit identifier yet a combination of non-sensitive attributes still uniquely points to an individual. Placeholder substitution also destroys task-critical context, and restoration depends on a local mapping table, failing for semantically equivalent but literally different content. Second, differential-privacy perturbation~\cite{mai2023split,shi2022selective} protects privacy by injecting noise, which necessarily harms the semantic coherence of downstream reasoning, with an inherent privacy-utility trade-off~\cite{duchi2013local}; text-level DP rewriting (e.g., SANTEXT~\cite{yue2021santext}) is subject to the same trade-off. Third, the recent skill-distillation method P2Skill~\cite{ryu2026p2skill} has a local SLM perform decomposition, routing, paraphrasing, and reconstruction while a cloud model guides skill refinement in a loop. P2Skill is closest to our problem, but has two key gaps:

\begin{itemize}
\item \textbf{(i) Cloud dependence during refinement.} Verbatim verification (\cite{ryu2026p2skill}, \S3.3, ``Iterative refinement loop'', steps 2--3) shows that in P2Skill's refinement loop the cloud LLM, in its \emph{reference} role, ``answers the original prompt directly'', and in its \emph{evaluator} role scores the triple $(x, y_s, y_r)$, where $x$ is the PII-bearing original input; its refinement corpus explicitly mixes in PII samples (\cite{ryu2026p2skill}, \S4.1). That is, \textbf{before skill refinement completes, PII-bearing refinement data must leave the device}. For fairness: during \emph{deployment}, P2Skill's outbound requests pass a deterministic identifier matcher and contain only sanitized content, and its supervisor role receives only aggregate statistics---its leakage point lies in the refinement phase, not the inference phase.
\item \textbf{(ii) Coverage and decision gaps.} Its deterministic matcher covers only declared identifiers, missing semantic/quasi-identifier leakage and context loss from over-redaction (\cite{ryu2026p2skill} \S4.4 itself attributes residual leakage to missed detections at the decomposition stage).
\end{itemize}

All three classes additionally share a third defect: sanitization pipelines are one-shot scripts lacking cross-instance reuse and self-improvement. Every new input restarts strategy design or iteration from scratch; the system cannot learn from past successes/failures nor converge over deployment time to fewer iterations and higher success rates.

\subsection{This Work}

We propose \textbf{DR-SL} (Dehydrate-Rehydrate with Self-Learning loop), addressing these gaps at three levels (system overview in Figure~\ref{fig:overview}):

\begin{figure}[t]
\centering
\includegraphics[width=\textwidth,trim={49.3pt 271.9pt 89.2pt 269.2pt},clip]{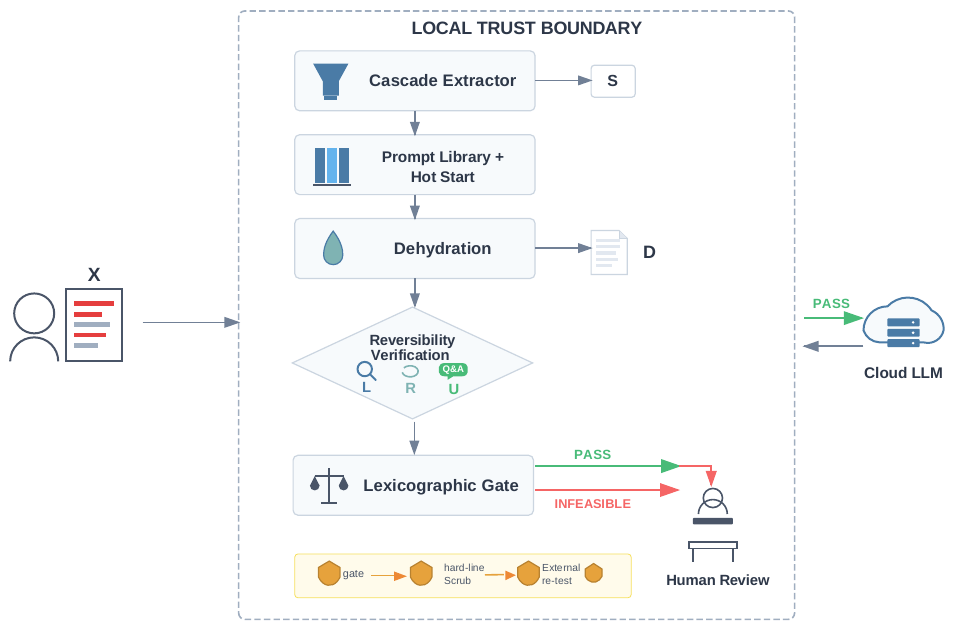}
\caption{DR-SL system overview. Inside the local trust boundary, the cascade extractor derives the sensitive set $S$ from the original material $X$; a type-matched prompt library hot-starts dehydration into the artifact $D$; the two-branch reversibility verification (re-identification probes $L$, restoration $R$, QA utility probes $U$) feeds the lexicographic gate; and the mandatory release chain (local gate $\to$ deterministic hard-line scrub $\to$ external strong-attacker re-test) decides egress: PASS instances send only $D$ to the cloud LLM, while INFEASIBLE instances emit a residual-leakage report and route to human review. The external re-test stage belongs to the destination trust class (same class as the destination cloud; trust semantics in \S\ref{sec:security})---only the dehydrated $D$ ever crosses the boundary.}
\label{fig:overview}
\end{figure}

\begin{enumerate}
\item \textbf{Decision level (reversibility verification).} De-identification completeness is formalized as two complementary, measurable conditions---\textbf{de-identification sufficiency} (Pufferfish semantics~\cite{kifer2012rigorous,kifer2014pufferfish}: the dehydrated artifact alone cannot support re-identification of sensitive information) and \textbf{task-information preservation} (measured by QA probes)---with the cascade extractor's recall explicitly delimiting the guarantee boundary. The decision mechanism runs entirely locally, driven by a local model and prompts; \textbf{nothing leaves the device before the completeness decision is made}. We state honestly: the leakage score produced by a local evaluator is a \textbf{lower-bound witness} of true leakage (Theorem~\ref{thm:fano})---a heuristic certificate, not an absolute upper bound; we strengthen it via offline empirical calibration (\S\ref{sec:security}) and a layered assurance model (\S\ref{sec:layered}).
\item \textbf{Mechanism level (self-learning loop).} Instance-level fixes are elevated to prompt-level global optimization through attribution, experience accumulation, prompt distillation, and versioned updates, with typed prompt libraries and hot-start, so that loop counts tend to decrease with accumulated experience (a design expectation, tested as hypothesis H1).
\item \textbf{Theory level (information-theoretic grounding).} A five-layer theoretical architecture characterizes the method's limits and boundaries---Pufferfish target semantics, Maximal Leakage as the attacker-independent limit~\cite{issa2020operational}, the rate-privacy frontier (utility-privacy Pareto limit), Wyner-Ziv conditional rate-distortion (fidelity limit), and an empirically calibrated approximation layer; we prove a leakage lower-bound theorem (Fano-type), a Pufferfish witness proposition, a feasibility criterion, and a termination theorem.
\end{enumerate}

\subsection{Contributions}

\begin{itemize}
\item A two-branch de-identification-completeness framework centered on \textbf{reversibility verification}, covering de-identification sufficiency (including semantic/quasi-identifier leakage) and task-information preservation (including over-redaction detection), together with a runtime cascade extractor for $S$ and a recall-ceiling analysis;
\item The \textbf{theoretical positioning} of the method: lower-boundedness of the leakage score (correcting a directional error), the Pufferfish witness, the rate-privacy feasibility criterion, termination of the lexicographic iteration, and the explicit stance of ``an empirically calibrated heuristic certificate approaching theoretical limits'';
\item A \textbf{fully local self-learning loop} (the offline stage touches only non-sensitive corpora) and a \textbf{layered assurance model} (soft certificate + deterministic hard line + human fallback), with a complete pre-registered experimental design.
\end{itemize}

\paragraph{Note on the method name.} DR-SL stands for ``Dehydrate-Rehydrate with Self-Learning loop''; throughout, the verified core of the method is the reversibility-verification loop, while the self-learning component is exploratory (\S6.9, T13/T17).

\subsection{Organization}

Section~2 surveys related work; Section~3 formalizes the problem, the runtime extraction of $S$, the two-branch reversibility verification, and theoretical properties; Section~4 presents the self-learning loop and an operationalized measurement framework for offline-online distribution shift (\S4.7); Section~5 gives the theoretical grounding and discussion; Section~6 presents the experimental design and results; Section~7 concludes. Appendix~A contains complete proofs; Appendix~B the proof-verification methodology.

\section{Related Work}

\subsection{Cloud-Local LLM Inference and Privacy Protection}

Cloud-local inference systems deploy a small model on-device and a large model in the cloud, routing simple or privacy-sensitive sub-tasks locally and forwarding complex reasoning to the cloud~\cite{hao2024hybrid,jin2025cecollm,li2025collab,siyan2025papillon}. Representative work includes hybrid SLM/LLM edge-cloud collaborative inference~\cite{hao2024hybrid}, the adaptive cloud-edge framework CE-COLLM~\cite{jin2025cecollm}, a survey of edge-SLM/cloud-LLM collaborative inference and learning~\cite{li2025collab}, and PAPILLON~\cite{siyan2025papillon}, which formulates Privacy-Conscious Delegation with a local privacy-aware proxy deciding which requests to forward---a canonical ``local routing'' design.

For privacy-preserving inference, cryptographic approaches~\cite{gilad2016cryptonets,mohassel2017secureml} offer strong guarantees but require specialized primitives and complex multi-party infrastructure, with high latency and deployment cost; NER-based de-identification~\cite{presidio,rehydra} strips identifiable spans but loses task-critical context; local differential privacy (LDP)~\cite{mai2023split,shi2022selective,duchi2013local} perturbs tokens before transmission, but noise harms coherence and the privacy-utility trade-off is inherent; PrivacyRestore~\cite{zeng2025privacyrestore} removes privacy spans client-side and restores them server-side via meta-vectors---a ``server-side restoration'' paradigm complementary to, but different from, our ``decide-before-egress'' local loop; its restoration relies on cloud processing of meta-vectors, whose own sensitivity warrants assessment (this is our inference; see the table note in \S2.7). Overall, privacy modules are trained as separate components requiring per-deployment adaptation~\cite{ryu2026p2skill}, and a unified ``is de-identification complete?'' decision mechanism is missing.

\subsection{The Sanitize-Rehydrate Paradigm and Reversible Anonymization}

Reversible anonymization is an emerging direction, represented by Rehydra~\cite{rehydra} and Presidio~\cite{presidio}. Their strengths are reversibility, conversational coherence, and session-level identity persistence; but they are essentially per-entity substitution, unprotected against uncovered sensitive types and semantic leakage, and restoration depends on local mapping tables, failing for semantically equivalent but literally different content.

\subsection{P2Skill and Skill Distillation}

P2Skill~\cite{ryu2026p2skill} has a frozen local SLM perform decomposition, PII-aware routing, paraphrasing, and reconstruction via four skill prompts, while a cloud model plays reference, evaluator, and supervisor roles in an iterative loop, distilling ``privacy skills'' into prompts rather than weights. Its merits: no fine-tuning, no learned auxiliary detectors, and a deterministic identifier matcher as the final gate for cloud-bound content; on 160 PRISM prompts (medical, banking, tourism, general knowledge) it improves privacy-preserved inference quality by $1.69\times$ and $3.66\times$ over baselines. As stated in \S1.2 (verbatim-verified and scope-limited): during \textbf{skill refinement}, its reference/evaluator branches require the cloud to read original inputs (\cite{ryu2026p2skill} \S3.3, steps 2--3), constituting a refinement-phase leakage point; deployment-phase egress is gated by the matcher; the supervisor receives only aggregate statistics. Its matcher also covers only declared identifiers, and its self-reported residual leakage stems mainly from missed detections at decomposition (\cite{ryu2026p2skill} \S4.4)---mutually corroborating our recall-ceiling analysis in \S3.2.

\subsection{Information-Theoretic Quantification and Privacy Measures}

Mutual information (MI) is the classic framework for measuring privacy leakage~\cite{belghazi2018mine,heerklotz2025neural}; MINE~\cite{belghazi2018mine} estimates high-dimensional MI with neural networks; the privacy funnel~\cite{eswa2022funnel} optimizes the privacy-utility trade-off by minimizing $I(\text{sensitive};\text{output})$ and maximizing $I(\text{task};\text{output})$. Neural MI estimation requires training and per-dataset iteration, unsuitable for real-time per-instance decisions. Closer to this work are two operational theoretical tools: \textbf{Pufferfish privacy}~\cite{kifer2012rigorous,kifer2014pufferfish} explicitly declares a secret set and discriminative pairs, bounding how much released data can amplify an attacker's posterior odds---fitting ``inference protection over declared secrets''; \textbf{Maximal Leakage}~\cite{issa2020operational} is the supremum, over all randomized functions of the source, of the multiplicative increase in guessing probability upon observing the release---attacker-independent, giving a theoretical ceiling on leakage. We anchor our theory in these two and use reversibility verification as the empirical approximation.

\subsection{Utility Evaluation for Text Sanitization}

For utility evaluation, Text Preserved Similarity (TPS)~\cite{pilan2025truthful} defines sanitization utility as preserved task-relevant information, aligning with our task-information-preservation goal; BERTScore~\cite{zhang2020bertscore} and TAE~\cite{tae2026} provide semantic-similarity measures. QA-based evaluation (e.g., QAFactEval~\cite{fabbri2022qafacteval}) measures information preservation by ``can task questions be answered from the text''---the methodological source of our QA probe (\S3.3). We take the QA probe as the primary utility metric, the rehydration score as a secondary diagnostic, and validate metric effectiveness with human-substitute cross-model studies.

\subsection{Self-Learning and Prompt Optimization}

Related to our self-learning mechanism are self-refinement methods and automatic prompt optimization, whose shared challenge is preventing prompts from overfitting training/validation cases. Our memorization check (\S4.3) and independent generalization scoring target exactly this problem.

\subsection{Positioning}

\begin{table}[H]
\centering
\caption{Comparison of representative methods. Each cell is marked ``fact'' (from the cited source) or ``our inference/claim''.}
\small
\begin{tabular}{@{}p{2.6cm}p{3.1cm}p{2.5cm}p{4.6cm}p{2.2cm}@{}}
\toprule
Method & Leakage coverage & Over-redaction detection & Evaluation/refinement leakage & Cross-instance reuse \\
\midrule
Placeholder substitution~\cite{presidio,rehydra} & Declared identifiers (fact) & None (fact) & None (fact) & None (fact) \\
LDP perturbation~\cite{mai2023split,shi2022selective} & Probabilistic (fact) & None (inherent noise loss, fact) & None (fact) & None (fact) \\
PrivacyRestore~\cite{zeng2025privacyrestore} & Declared identifiers (fact) & Partial (restoration fidelity, fact) & \textbf{Yes} (cloud processes meta-vectors; our inference) & None (fact) \\
P2Skill~\cite{ryu2026p2skill} & Declared identifiers + rewriting (fact) & Partial (reconstruction skill, fact) & \textbf{Yes}: during refinement, cloud reference/evaluator branches read original inputs (fact, \cite{ryu2026p2skill} \S3.3 steps 2--3); none at deployment (fact) & Limited (skill distillation, fact) \\
\textbf{DR-SL (this work)} & \textbf{Declared + quasi-identifiers within extractor recall} (our claim, conditioned on recall) & \textbf{Yes} (task-information-preservation metric) & \textbf{None} (offline: non-sensitive corpora only; online fully local) & \textbf{Yes} (prompt library + hot start) \\
\bottomrule
\end{tabular}
\end{table}

\section{Problem Definition and Method}

\subsection{Problem Formalization}

Let the original material be a random variable $X$, the sensitive set $S=\{s_1,\dots,s_n\}$ (each $s_i$ a sensitive attribute with domain $V_i$; explicit identifiers or quasi-identifiers), the attacker's auxiliary knowledge $A$, and the task-relevant variable $T$. The dehydrator $\mathcal{D}$ maps $X$ to the dehydrated artifact $D=\mathcal{D}(X)$. At runtime $S$ is produced by the cascade extractor (\S3.2). Target conditions:

\paragraph{De-identification sufficiency (Pufferfish semantics).}
We adopt the Pufferfish framework~\cite{kifer2012rigorous,kifer2014pufferfish}: the secret set $\mathcal{S}$ is induced by $S$---for each $s_i$ and value $v\in V_i$, the secret proposition $\sigma_{i,v} := \text{``}s_i = v\text{''}$; the discriminative-pair set $\mathcal{Q}$ contains all same-attribute distinct-value pairs $\{(\sigma_{i,v},\sigma_{i,v'})\}$ and identity pairs (a quasi-identifier combination pointing to a unique individual); $\Theta$ is the attacker's prior-belief set. The dehydration mechanism $P(D\mid X)$ satisfies $(\varepsilon,\mathcal{Q},\Theta)$-Pufferfish de-identification sufficiency iff for all $(\sigma,\sigma')\in\mathcal{Q}$, $\theta\in\Theta$, and outputs $d$:
\begin{equation}
e^{-\varepsilon} \;\le\; \frac{P(D=d\mid \sigma,\theta)}{P(D=d\mid \sigma',\theta)} \;\le\; e^{\varepsilon}. \label{eq:pufferfish}
\end{equation}
Semantics: releasing $D$ amplifies the attacker's prior odds on any discriminative pair by at most $e^{\varepsilon}$. Unlike differential privacy, Pufferfish directly constrains inference about declared secrets without adding noise---precisely what the re-identification branch aims to prevent. The conditional probabilities in \eqref{eq:pufferfish} are relative to attacker knowledge $A$ (operationalized as graded attackers in \S3.3).

\paragraph{Task-information preservation.}
The fidelity goal changes from ``restore the original text'' to ``preserve task-relevant information'': maximize $I(D;T)$. Its operational measure is the QA-probe utility score $U$ (\S3.3). The original rehydration score $R$ (similarity of $\hat X$ reconstructed from $(D,S)$ to $X$) is demoted to a secondary diagnostic signal.

\paragraph{Notation.}
Normalized mutual information $I_{\mathrm{norm}}(D;S) := I(D;S)/H(S) \in [0,1]$; binary entropy $H_2(\cdot)$. The mapping between the information-theoretic threshold $\varepsilon$ and the operational threshold $\tau_L$ is given by Theorem~\ref{thm:fano} ($\tau_L$ corresponds to a leakage \emph{lower-bound} guarantee).

\begin{definition}[Reversible compression]
If there exists a map $f$ with $f(D,S)=X$, $D$ is a reversible compression of $X$ with respect to $S$. This definition is retained as the theoretical object of the rehydration branch (diagnostic).
\end{definition}

\begin{definition}[De-identification completeness]
$\mathcal{D}$'s de-identification of $(X,S)$ is complete iff $D=\mathcal{D}(X)$ satisfies: (i) the operational criteria $L \le \tau_L$ (in the Wilson-interval upper-bound sense, \S3.3) and $U \ge \tau_U$; (ii) these criteria hold under extractor recall $\hat\rho$ and attacker class $A_3$ (conditionality of the guarantee in \S\ref{sec:security}).
\end{definition}

\subsection{Runtime Extraction of the Sensitive Set $S$ (Cascade Extractor)}

The key premise raised in review---``where does $S$ come from at runtime?''---is answered by a cascade extractor:
\begin{itemize}
\item \textbf{(a) Regex/NER high-precision layer}: explicit identifiers (ID numbers, phone numbers, emails, names, accounts), with finite pattern classes and near-complete recall;
\item \textbf{(b) Local-LLM high-recall layer}: quasi-identifiers (age, department, surgical history, and normalizations), via an extraction prompt plus a self-critique second pass dedicated to missed items;
\item \textbf{(c) Combination layer} (DetectCombinable): decides which combinations of extracted attributes constitute quasi-identifiers, emitting candidates for Branch-1 combination questions (\S3.3) and the synthetic-population $k$ test (\S6).
\end{itemize}

\textbf{Extraction recall}: on a gold-annotated set, $\mathrm{Recall} = |S_{\mathrm{found}} \cap S_{\mathrm{gold}}|\,/\,|S_{\mathrm{gold}}|$, reported stratified (explicit / quasi / overall).

\begin{proposition}[Recall ceiling of system leakage, over-extraction-corrected]\label{prop:recall}
Let $G=S_{\mathrm{gold}}$ and $F=S_{\mathrm{found}}$. Branch~1 probes only attributes in $F$; missed attributes are never probed. Writing $L_{\mathrm{measured}}$ for the hit rate over the probed set and $\mathrm{Recall}=|F\cap G|/|G|$, the worst-case decomposition is
\[
L_{\mathrm{system}} \;\le\; \frac{|F|}{|G|}\,L_{\mathrm{measured}} \;+\; \bigl(1 - \mathrm{Recall}\bigr).
\]
When the extractor over-extracts ($|F|>|G|$; in our data $|F|/|G|$ averages 1.43), the first term carries an \emph{inflation factor}: spurious probes that miss dilute $L_{\mathrm{measured}}$, so the true-attribute hit rate can be up to $|F|/|G|$ times larger.
\end{proposition}

Hence every system-level safety conclusion must be conditioned on the measured recall $\hat\rho$ \emph{and} the over-extraction ratio: ``given extractor recall $\ge \hat\rho$ and $|F|/|G| \le \iota$, system residual leakage $\le \iota\cdot L_{\mathrm{measured}} + (1-\hat\rho)$''. On our Stage-1 data ($\hat\rho=0.951$, mean $\iota=1.43$): residual $\le 1.43\,L_{\mathrm{measured}} + 0.049$. Theoretical limits: explicit-identifier recall can approach 1 (finite pattern class); quasi-identifier recall has no finite upper bound---whether an attribute is identifying depends on the attacker's auxiliary knowledge $A$, so recall can only be defined relative to a given $A$ as $\mathrm{Recall}(A)$. \S3.3/\S6 therefore report recall and leakage curves over graded attackers $A_0$--$A_3$.

\subsection{Two-Branch Reversibility Verification}

\paragraph{Branch 1 (re-identification branch / de-identification sufficiency).}
The evaluator receives only the dehydrated artifact $D$ and general knowledge (plus graded auxiliary knowledge), and answers a question set $Q$ constructed from $S$ (Algorithm~1). Probing each attribute $m$ times yields per-attribute hit rates $L_i$ and the overall leakage score
\[
L = \frac{1}{|S|}\sum_i L_i \in [0,1].
\]

\textbf{Statistical protocol.} Per-attribute probes $m = \lceil 75/|S| \rceil$ (total $\ge 75$, ensuring the gate is passable even at zero hits); $L$ is reported as the mean over probes, and the decision uses a \textbf{Wilson confidence interval}~\cite{wilson1927}: pass iff the 95\% Wilson upper bound of $L$ is $\le \tau_L$. $\tau_L$ thus means ``a threshold on the interval upper bound of the hit rate'', resolving excessive binomial noise at small $|S|$ (e.g., at $|S|=5$ a single-point granularity of 0.2 would otherwise be equivalent to zero tolerance).

\paragraph{Graded attackers (operationalizing $A$).}
The ``supremum over $A$'' is operationalized as a maximum over a finite nested attacker set:

\begin{table}[H]
\centering\small
\begin{tabular}{@{}llp{8.2cm}@{}}
\toprule
Level & Auxiliary knowledge & Implementation \\
\midrule
$A_0$ & General knowledge only & Bare local re-identification \\
$A_1$ & Domain statistical priors & Re-identification prompt injected with domain frequency knowledge \\
$A_2$ & Public-data retrieval & Re-identification branch augmented with public-corpus RAG \\
$A_3$ & Linkage-attack database & Synthetic-population $k$ counting (no LLM) \\
\bottomrule
\end{tabular}
\end{table}

Knowledge sets satisfy $A_0 \subseteq A_1 \subseteq A_2 \subseteq A_3$, so the $L(A_i)$ curve should be monotone; non-monotonicity is treated as evaluation instability and triggers review. The system reports the four-level leakage curve (x-axis: attacker level; y-axis: per-instance mean hit rate $\pm$ 95\% Wilson interval, by domain).

\paragraph{$A_3$ k-anonymity test~\cite{sweeney2002k}.}
For the quasi-identifier combination $q$ extracted from $D$, count matching records in a synthetic population (\S6; the controlled world reuses the dataset generator's distribution, the real-statistics world is built from public marginal distributions via iterative proportional fitting (IPF)~\cite{deming1940ipf}; fully synthetic and reproducible):
\[
k(q) = \bigl|\{\, r \in \mathrm{DB} : r\text{'s quasi-identifier combination matches } q \,\}\bigr|.
\]
Semantics: the attacker can at best narrow the target to a set of $k$ individuals; $k=1$ means unique identification. The $A_3$ hit event is $k \le k^{*}$ ($k^{*}=5$ frozen in the protocol), unifying semantics across the four-level curve.

\paragraph{Branch 2 (rehydration branch / diagnostic).}
The evaluator receives $(D,S)$ and reconstructs $\hat X$; $\mathrm{Sim}(\hat X, X)$ (literal matching primary, BERTScore~\cite{zhang2020bertscore} fallback) gives the restoration score $R$. $R$ is demoted to a secondary diagnostic: it detects rehydrator hallucination and over-redaction but is no longer the primary fidelity criterion.

\paragraph{Primary utility metric: QA probes~\cite{pilan2025truthful,fabbri2022qafacteval}.}
Fully decoupled from restoring the original text; directly measures ``task-information availability from the cloud's viewpoint'':
\begin{enumerate}
\item Generate a task-relevant question set $Q_T$ from $X$ plus the task description ($n_q$ questions, not targeting sensitive attributes themselves; filtered and spot-checked);
\item Reference answers $A_{\mathrm{ref}} = \mathrm{ans}(q\mid X)$ (fixed reference model; self-consistency sampling for open questions; known ground truth for structured ones);
\item Target answers $A_D = \mathrm{ans}(q\mid D)$ (same model, same sampling parameters---exactly the deployment condition);
\item Per-item scoring $\mathrm{Agree}_j$: structured items by exact/tolerance match $\{0,1\}$; open items by judge rubric \{correct 1, partial 0.5, wrong 0, contradicted 0\}, the judge seeing only $(q, A_{\mathrm{ref}}, A_D)$; ``cannot determine'' counts 0.25 (frozen);
\item Utility $U = \frac{1}{|Q_T|}\sum_j \mathrm{Agree}_j$, reported as per-instance/per-domain means with bootstrap confidence intervals.
\end{enumerate}

\textbf{Completeness criterion (operational).} Dehydration succeeds $\iff$ WilsonUpperBound$(L) \le \tau_L$ and $U \ge \tau_U$ (lexicographic; Algorithm~2$'$).

\subsection{Question-Set Generation}

\textbf{Algorithm 1 (QuestionSetGen)}: input $S$ and a typed template library $T$; instantiate template questions per attribute by type; generate ``can the individual be uniquely identified?'' questions for DetectCombinable outputs; deduplicate and return $Q$.

\subsection{Lexicographic Iteration and Convergence Criteria}

The bidirectional same-round update of v1.0 could oscillate without termination guarantees. We switch to \textbf{lexicographic optimization}: utility as a hard gate, leakage as a soft objective, single-attribute strengthening, single-step rollback.

\begin{algorithm}[H]
\caption{DR-SL per-instance main loop (lexicographic).}
\label{alg:main}
\KwIn{original material $X$, sensitive set $S$ (extractor output, recall $\hat\rho$), thresholds $\tau_L,\tau_U$, round budget $B=3|S|+5$, stability window $N$, prompt library $P$, local evaluator $M$}
\KwOut{PASS($D^{*}$) or INFEASIBLE(residual-leak report)}
$P_{\mathrm{init}} \leftarrow \mathrm{HotStart}(P,S)$\tcp*{type-matched hot start}
$D \leftarrow \mathrm{Dehydrate}(X,S,P_{\mathrm{init}})$\;
$Q \leftarrow \mathrm{QuestionSetGen}(S)$\tcp*{Algorithm~1}
$(L,CI,U) \leftarrow \mathrm{Evaluate}(M,D,Q,S,X)$\tcp*{$L$ with Wilson interval; $U$ = QA probe}
\While(\tcp*[h]{Phase 1: satisfy the utility gate first}){$U<\tau_U$ \textbf{and} $t<B$}{
  $D \leftarrow \mathrm{LoosenDehydrate}(D,X,S,P)$\tcp*{loosen only; no leakage strengthening}
  exit phase if two consecutive loosens fail to improve $U$ (no-progress)\;
  $(L,CI,U) \leftarrow \mathrm{Evaluate}(M,D,Q,S,X)$\;
}
$\mathrm{skipped}\leftarrow\emptyset$;\quad $\mathrm{triedLoosen}\leftarrow\emptyset$\;
\While(\tcp*[h]{Pass := Wilson upper $\le\tau_L$ and $U\ge\tau_U$}){$\neg\mathrm{Pass}(L,CI,\tau_L,U,\tau_U)$ \textbf{and} $t<B$}{
  $s^{*} \leftarrow \argmax_{s_i\in S\setminus\mathrm{skipped}} \mathrm{LeakContrib}(s_i,D)$\;
  $D' \leftarrow \mathrm{StrengthenAttribute}(D,s^{*},X,S,P)$\tcp*{single-attribute step}
  $(L',CI',U') \leftarrow \mathrm{Evaluate}(M,D',Q,S,X)$\;
  \uIf(\tcp*[h]{hard gate: utility never regresses}){$U'\ge\tau_U$}{
    $D\leftarrow D'$; accept; update stability window\;
  }
  \uElseIf(\tcp*[h]{compensatory loosen, one chance}){$s^{*}\notin\mathrm{triedLoosen}$}{
    $\mathrm{triedLoosen}\leftarrow\mathrm{triedLoosen}\cup\{s^{*}\}$\;
    $D'' \leftarrow \mathrm{LoosenDehydrate}(D,X,S,P)$\;
    \If{$U(D'')>U$}{
      $D\leftarrow D''$; $\mathrm{triedLoosen}\leftarrow\emptyset$; \textbf{continue}\;
    }
    $\mathrm{skipped}\leftarrow\mathrm{skipped}\cup\{s^{*}\}$\tcp*{not strengthenable within gate $\to$ residual}
  }
  \Else{
    $\mathrm{skipped}\leftarrow\mathrm{skipped}\cup\{s^{*}\}$; attribute failure $\to$ experience store\;
  }
  $t\leftarrow t+1$\;
}
\eIf{$\mathrm{Pass}(L,CI,\tau_L,U,\tau_U)$ \textbf{and} window stable}{
  \Return PASS($D$)\tcp*{egress; audit log retained}
}{
  \Return INFEASIBLE($\mathrm{skipped}\cup{}$failing attributes, with hit-question details)\tcp*{to human review; egress refused}
}
\end{algorithm}

Key points: every accepted step maintains $U \ge \tau_U$ (gate invariant; the fidelity criterion \emph{is} task-information preservation, with the rehydration score $R$ diagnostic only); after a successful compensatory loosen, \texttt{triedLoosen} resets (new headroom available); skipped attributes permanently enter the residual set, avoiding infinite retries; consecutive no-improvement loosens terminate Phase~1; the hard round budget $B$ guarantees termination (Theorem~\ref{thm:termination}).

\begin{definition}[Stable pass]
Over window $W=\{t-N+1,\dots,t\}$: max CI upper bound $\le \tau_L$, $\min U_w \ge \tau_U$, and fluctuations of $L$ (and diagnostic $R$) $\le \delta$.
\end{definition}

\textbf{Differences from v1.0}: same-round bidirectional updates removed; the utility gate admits only Phase-1 loosening and non-regression thereafter; the loop has a hard budget $B$; the failure path is a formal INFEASIBLE branch (residual-leakage report), not a deadlock.

\subsection{Theoretical Properties (Graded)}

Grading rule: \textbf{Theorem} = full proof (appendix); \textbf{Proposition} = proof sketch in text, full proof in appendix; \textbf{Hypothesis} = to be tested experimentally.

\begin{proposition}[Directional correctness of the two-branch alarm]\label{prop:direction}
Fix the evaluator's capability. (i) If $I(D;S)=0$ ($D$ independent of $S$), then for any template family $\mathcal{Q}$, Branch-1 hit rate does not exceed the prior guessing rate (no information gain); (ii) if $H(X\mid D,S) > \eta$, then any rehydrator has expected restoration score $\mathbb{E}[R] < 1$. Each branch thus alarms in the correct direction when its sufficiency condition is violated.
\end{proposition}

\begin{theorem}[Lower-boundedness of the leakage score, Fano-type]\label{thm:fano}
Let an attacker (including the local evaluator) achieve per-attribute hit rate $L_i$ on $s_i$ (domain size $|V_i|$), with $L_i \ge 1/|V_i|$. Then
\[
I(D;S) \;\ge\; H(S) - \sum_i \bigl[\, H_2(1-L_i) + (1-L_i)\log(|V_i|-1) \,\bigr],
\]
and hence $I_{\mathrm{norm}}(D;S) \ge 1 - \sum_i[\cdot]/H(S)$.
\end{theorem}

\emph{Interpretation} (directional correction, and its limits): measured attack success gives a \textbf{lower bound} on true leakage; stronger attackers can only raise the bound. A local evaluator's failure to re-identify \textbf{cannot} prove that no attacker can. \textbf{The bound certifies leakage, never safety}---this one-sidedness is structural to any Fano-type argument, and we quantify how weak the guarantee is at our operating point: with $\tau_L=0.05$ and a uniform attribute over $|V|=100$ values, passing the gate certifies only $I(D;s) \ge H(s) - [H_2(0.95) + 0.95\log_2 99] \approx 6.64 - 6.59 \approx 0.06$ bits ($\approx$1\% of $H(s)$); for open-domain quasi-identifiers ($|V| \sim 10^4$, e.g., hospital names) the bound is $\approx 0.34$ bits, $I_{\mathrm{norm}} \ge 0.026$---essentially vacuous. \textbf{The operative safety argument of this paper therefore does not rest on Theorem~\ref{thm:fano}}: it rests on (i) the empirical calibration $\hat\Delta$ (\S\ref{sec:security}), (ii) the deterministic hard line (all declared identifiers scrubbed before egress, independent of any model), and (iii) human fallback on INFEASIBLE. The theorem's role is to \emph{witness leakage for rejection and warning} (directions where a lower bound is the useful side), and to make the certificate's evaluator-dependence explicit. The certificate strength increases monotonically with evaluator capability. Full proof: Appendix~A.2.

\emph{Conditioning on auxiliary knowledge.} When attackers use auxiliary knowledge $A$ (\S3.3), all entropy terms are understood conditional on $A$: the theorem applies verbatim to $I(D;S\mid A)$ and $H(S\mid D,A)$, with hit rates measured by $A$-equipped attackers; we omit the conditioning in the statement for readability.

\begin{proposition}[Pufferfish witness: hit rate $\to$ $\varepsilon$ lower bound]\label{prop:witness}
Let the optimal attacker's hit rate on attribute $s_i$ be $L$, with prior $p_{\max} = \max_v P(s_i{=}v)$, $p_{\min} = \min_v P(s_i{=}v)$. Then there exists a discriminative pair $(\sigma^{*},\sigma')$ such that any mechanism satisfying \eqref{eq:pufferfish} must have
\[
\varepsilon \;\ge\; \ln\!\Bigl(\frac{L}{1-L}\Bigr) + \ln\!\Bigl(\frac{p_{\min}}{p_{\max}}\Bigr);
\qquad\text{uniform prior: } \varepsilon \ge \ln\!\Bigl(\frac{L}{1-L}\Bigr).
\]
\end{proposition}

\emph{Interpretation}: attacks give a lower bound on the mechanism's required $\varepsilon$ (a leakage witness), same direction as Theorem~\ref{thm:fano}. Numerically: $L=0.9 \Rightarrow \varepsilon \ge 2.2$ (uniform prior); $L=0.05$ gives only a loose bound---the certificate is strong only when the evaluator is strong. Full proof: Appendix~A.3.

\begin{proposition}[Feasibility criterion, rate-privacy frontier]\label{prop:feasibility}
For any dehydrator, task-information preservation satisfies
\[
I(D;T) \;\le\; I(D;S) + I(X;T\mid S),
\]
so the rate-privacy function $U^{*}(\varepsilon) := \max_{p(D|X):\, I(D;S|A)\le\varepsilon} I(D;T)$ satisfies $U^{*}(\varepsilon) \le \varepsilon + I(X;T\mid S)$.
\end{proposition}

\emph{Interpretation}: if the utility requirement satisfies $\tau_U > \tau_L + I(X;T\mid S)$, the $(\tau_L,\tau_U)$ feasible region is empty and Algorithm~2$'$ necessarily outputs INFEASIBLE. When task information is strongly coupled to sensitive information (the medical example: removing age/surgical history makes advice worthless; $I(X;T|S)$ small), the feasible region shrinks or empties---\textbf{infeasibility is a property of the problem, not a defect of the method}. Full proof: Appendix~A.4.

\begin{theorem}[Termination of the lexicographic iteration]\label{thm:termination}
Algorithm~2$'$ terminates in finitely many rounds, outputting PASS or INFEASIBLE; and after Phase~1, every accepted step maintains $U \ge \tau_U$ (gate invariant).
\end{theorem}

\begin{hypothesis}[Loop-count decrease; H1]\label{hyp:h1}
Let $n_c^{(v)}$ be the expected initial iteration rounds for type-$c$ material under prompt version $v$. If distilled rules eliminate generalizable failure modes for type $c$, then $n_c^{(v+1)} \le n_c^{(v)}$, and expected total iterations decrease across versions. (v1.0 stated this as a proposition with the memorization check ``guaranteeing generalization''---too strong; it is now a hypothesis tested in \S6. The memorization check only prevents memorizing concrete instances.)
\end{hypothesis}

\section{The Self-Learning Loop}

\subsection{Two Levels of Optimization Objects}

Iteration involves two kinds of optimization objects with vastly different reuse value: \emph{instance-level} fixes (scoped to one material; no cross-instance reuse) and \emph{prompt-level} rules distilled from repeated failures (global scope; reusable). DR-SL's core principle: \textbf{elevate instance-level fixes into prompt-level optimization}. Without this, the system degenerates into per-instance manual tuning and loop counts do not fall.

\subsection{The Attribution--Accumulation--Distillation--Update Loop}

\begin{enumerate}
\item \textbf{Attribution}: on each abnormal branch feedback, besides fixing the current instance, locate the root cause---which sensitive type or context structure (lists, dialogues, long text) the dehydration prompt failed to cover---emitting a structured (material, cause, action) triple.
\item \textbf{Accumulation}: failure cases + root causes + fix actions are persisted in a local experience store, indexed by domain/type/failure mode, searchable and auditable.
\item \textbf{Distillation}: periodically (every $K$ entries or daily), the local model induces general rules from the store, subject to the \textbf{memorization check}---rejecting rules that embed concrete instance values, keeping only abstract patterns.
\item \textbf{Update}: accepted rules merge into dehydration/rehydration/evaluation prompts as new versions with full version history and rollback support. Prompts become the sole carrier of ``de-identification knowledge''.
\end{enumerate}

\subsection{Memorization Check}

\textbf{Definition~4 (Memorization check)} and Algorithm~3 as in v1.0 (literal/normalized value check + quote-length threshold $\ell$ + generalization-score threshold $\theta_g$), with two corrections:
\begin{itemize}
\item \textbf{Scoped responsibility}: the memorization check intercepts only literal/quotation leakage; it \textbf{no longer claims to ``guarantee generalization''} (v1.0 overclaim);
\item \textbf{Independent GenScore}: the generalization score $\theta_g$ is produced by an \textbf{independent scorer}---an offline strong model (on non-sensitive corpora) or a local model from a different family than the distiller; self-scoring is forbidden, eliminating circular reasoning. Formally, a candidate rule $r$ is rejected iff
\[
\exists s \in S:\ s \in r \;\;\lor\;\; \exists c \subseteq X,\ |c|>\ell:\ c \in r \;\;\lor\;\; \mathrm{GenScore}_{\mathrm{indep}}(r)<\theta_g .
\]
\end{itemize}

\subsection{Prompt Library and Hot-Start Reuse}

A layered prompt system: \emph{global baseline prompts} (universal rules), \emph{typed prompt libraries} (medical/financial/credential/dialogue, since sensitive points and strategies differ markedly by type), and \emph{hot start} (new materials initialize from the mature prompt of their type rather than cold-starting). The loop-count-decrease prediction is demoted to \textbf{Hypothesis H1} (\S3.6), tested in \S6.9 (T13).

\subsection{Overall Architecture}

\begin{figure}[H]
\centering
\includegraphics[width=0.85\textwidth,trim={25.2pt 148.8pt 24.8pt 149.2pt},clip]{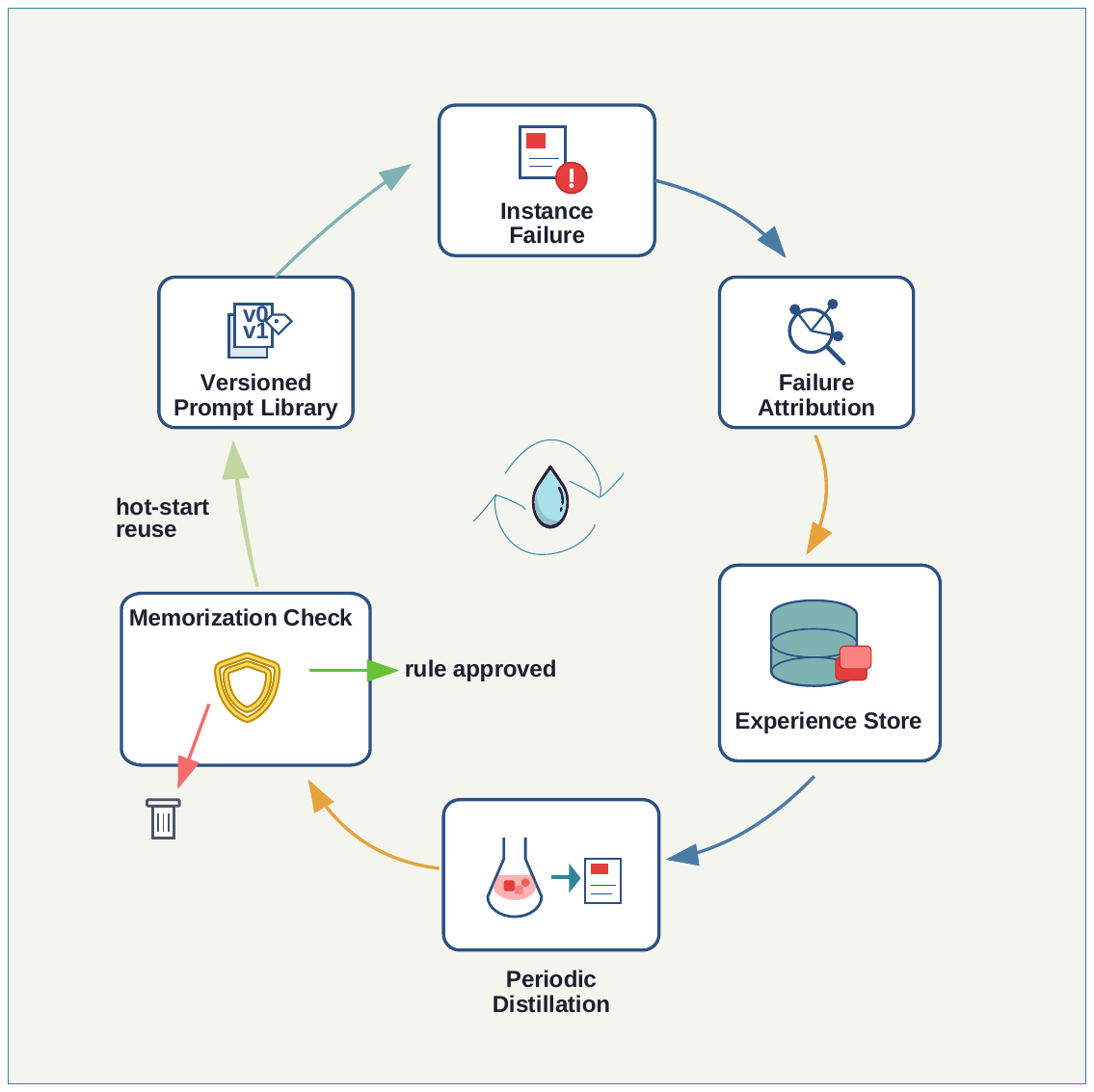}
\caption{Self-learning loop architecture: instance failures are attributed, accumulated in the experience store, periodically distilled into rules, gated by the memorization check, and committed as versioned prompt-library updates that hot-start subsequent instances. The per-instance main loop and its PASS/INFEASIBLE decision exits are specified in Algorithm~\ref{alg:main}.}
\label{fig:looparch}
\end{figure}

\subsection{Two Boundaries}

\begin{itemize}
\item \textbf{Generalization ceiling}: uncovered new types require cold starts; on convergence failure the system identifies a new type and triggers distillation. Loop-count decrease is asymptotic, not one-shot.
\item \textbf{Overfitting risk}: mitigations are the memorization check, distillation from diverse cases, independent generalization scoring, and version rollback.
\end{itemize}

\subsection{Offline-Online Distribution Shift: An Operationalized Measurement Framework}\label{sec:shift}

Prompt libraries distilled offline face domain shift relative to the online distribution. This section gives the physical meaning of each theoretical quantity and frozen measurement recipes, turning ``the generalization problem'' into a repeatably measurable object rather than a verbal concern.

\paragraph{Key claim (safety-efficiency separation).}
Distribution shift affects only \emph{efficiency} (iteration rounds) and \emph{utility}, not \emph{safety decisions}---the release decision (re-identification branch + deterministic hard line) is computed online per current instance; the offline prompt library only initializes and never participates in release decisions. The risk of shift is therefore ``bounded cost'', not ``a safety breach''. We quantify the efficiency side.

\paragraph{Theoretical hook (domain-adaptation bound~\cite{bendavid2010theory}).}
Let $S$ be the offline distillation distribution, $T$ the online distribution, and $h$ the prompt policy. Then
\begin{equation}
\varepsilon_T(h) \;\le\; \varepsilon_S(h) + \tfrac12\, d_{\mathcal{H}\Delta\mathcal{H}}(S,T) + \lambda(S,T). \label{eq:da}
\end{equation}
Physical meaning and measurement of each term:
\begin{enumerate}
\item $\varepsilon_T(h)$ (\textbf{online expected loss}): expected loss when the offline-distilled prompt library serves requests drawn from the online distribution. The loss is instantiated per instance (two frozen forms, both reported): $\ell_1$ = normalized rounds (actual/budget $B$); $\ell_2 = 1-\mathrm{Pass}$ (0-1 convergence failure).
\item $\varepsilon_S(h)$ (\textbf{offline error}): mean loss of the same pipeline, same prompts, same loss definition on the offline held-out corpus (calibration split, non-sensitive by construction). The dataset splits support this natively: calibration$\to$train act as ``historical traffic'' (visible offline), test(+adversarial) as ``online traffic''---$\Delta$ is measurable without real deployment.
\item $d_{\mathcal{H}\Delta\mathcal{H}}$ (\textbf{computable proxy}): the sup over the full function class is incomputable. We freeze a finite probe class $\mathcal{H}_{\mathrm{probe}}$ (TF-IDF + logistic regression; frozen sentence embeddings + logistic regression; LLM zero-shot discrimination) and compute $\hat d_{\mathcal{H}} = 2\max_{h'\in\mathcal{H}_{\mathrm{probe}}} |P_S(h'{=}1) - P_T(h'{=}1)|$. $\hat d_{\mathcal{H}}$ is a \textbf{lower-bound estimate} of the true divergence on the probe class and must not be cited as the true divergence.
\item $\lambda$ (\textbf{joint optimal error}): generally incomputable; handled by availability: (i) back-computed residual $\lambda \ge \Delta - \tfrac12 \hat d$ (primary); (ii) oracle approximation---distill a domain-specific prompt $h^{*}$ on pooled data, $\hat\lambda = \varepsilon_S(h^*)+\varepsilon_T(h^*)$; (iii) no independent claim otherwise.
\end{enumerate}

\paragraph{Three measurable proxies} (none touching sensitive values):
\begin{itemize}
\item \textbf{Domain-discriminator AUC}: sample $n{=}200$ from each of offline/online $\to$ TF-IDF $\to$ logistic regression, 5-fold CV $\to$ AUC~$\pm$~95\%~CI; $\hat d_{\mathcal{H}} \ge 2|\mathrm{AUC}-0.5|$. Reading: AUC~$\le$~0.6 negligible shift; 0.6--0.8 moderate; $>$0.8 severe.
\item \textbf{Embedding MMD}: $\mathrm{MMD}^2 = \mathbb{E}[k(x,x')] + \mathbb{E}[k(y,y')] - 2\mathbb{E}[k(x,y)]$ (Gaussian kernel), permutation test (1000 resamples) for the p-value; corroborating evidence for AUC.
\item \textbf{Type coverage}: instance signature $\mathrm{sig}(x)$ = set of sensitive types; library signature set $\mathcal{C}$; Coverage = fraction of online instances whose signature is exactly matched (primary) or subset-matched (optimistic) in $\mathcal{C}$. Zero model calls, most interpretable; the uncovered-signature list directly guides library expansion.
\end{itemize}

\paragraph{Dose-response experiment.}
Construct intermediate distributions $T_\alpha = (1-\alpha)S + \alpha T_{\mathrm{new}}$, with $T_{\mathrm{new}}$ a controlled new-type variant set (\S6.1) and $\alpha \in \{0,0.25,0.5,0.75,1\}$; measure $\Delta(\alpha)$ and $\hat d(S,T_\alpha)$ jointly and verify monotonicity (Spearman $\rho$ and sign test); report the full table $(\varepsilon_S,\varepsilon_T,\Delta,\mathrm{AUC},\mathrm{MMD}\ p,\text{coverage})$ plus the dose-response figure.

\section{Theoretical Grounding and Discussion}

\subsection{The Five-Layer Theoretical Architecture: From Semantics to Empirical Approximation}\label{sec:fivelayer}

\begin{table}[H]
\centering\small
\caption{Five-layer theoretical architecture.}
\begin{tabular}{@{}p{2.9cm}p{3.6cm}p{4.2cm}p{4.6cm}@{}}
\toprule
Layer & Tool & Quantity & Role \\
\midrule
L1 Target semantics & Pufferfish~\cite{kifer2012rigorous,kifer2014pufferfish} & posterior odds $\le e^{\varepsilon}$ & \textbf{Definition} of de-identification sufficiency (Eq.~\eqref{eq:pufferfish}) \\
L2 Attacker-independent limit & Maximal Leakage~\cite{issa2020operational} & $\mathrm{MaxL}(X\to D) = \log \sum_d \max_x P(d\mid x)$ & Theoretical ceiling of what Branch 1 measures; empirical $L$ can only be a lower-bound approximation \\
L3 Utility-privacy frontier & rate-privacy function~\cite{eswa2022funnel} & $U^{*}(\varepsilon) = \max I(D;T)$ s.t.\ $I(D;S\mid A)\le\varepsilon$ & Theoretical limit of information separation; feasibility criterion (Prop.~\ref{prop:feasibility}) \\
L4 Fidelity limit (secondary) & Wyner-Ziv conditional rate-distortion~\cite{wyner1976rate} & $R_{X|S}(\delta)$ & Theoretical positioning of the rehydration branch \\
L5 Empirical approximation & heuristic certificate + calibration & $L_{\mathrm{local}},\ \hat\Delta,\ \hat L = L_{\mathrm{local}}+\hat\Delta$ & The empirically calibrated certificate of \S\ref{sec:security} \\
\bottomrule
\end{tabular}
\end{table}

\textbf{Positioning statement}: DR-SL approaches the theoretical limit characterized by the rate-privacy frontier under Pufferfish semantics, via an empirically calibrated heuristic certificate; the closeness of approximation improves with evaluator capability (measured by $\hat\Delta$). We do not claim to reach the limit. \textbf{Honesty note on the architecture}: layers L2 (Maximal Leakage) and L4 (Wyner-Ziv) are \emph{positioning} layers in this paper---they delimit what the empirical quantities approximate but are not computed anywhere in our experiments (computable MaxL estimation is future work, \S7). The operative layers are L1 (semantics), L3 (feasibility criterion, used to interpret INFEASIBLE), and L5 (calibrated certificate). Compared with MINE~\cite{belghazi2018mine} / the privacy funnel~\cite{eswa2022funnel}, reversibility verification is training-free and usable online; its price is dependence on evaluator capability---exactly the trade-off characterized by Theorem~\ref{thm:fano} and empirical calibration.

\subsection{Security Model: A Two-Phase Threat Model and the Empirically Calibrated Certificate}\label{sec:security}

\paragraph{Two-phase threat model.}
\begin{itemize}
\item \textbf{Offline phase}: the cloud strong model acts \emph{only} on corpora containing \emph{no sensitive data} (public/synthetic; a data audit guarantees no real PII), serving as an \textbf{offline attacker benchmark and calibrator} for distilling the general prompt library and measuring evaluator bias. Non-sensitivity is guaranteed by construction, so the offline phase constitutes no leakage channel.
\item \textbf{Online phase}: all decisions and iterations are local. Attacker assumptions: (A1) obtains every egressed request $D$; (A2) has general domain knowledge and public corpora, but knows neither the local $S$ nor the experience store; (A3) can invoke arbitrarily strong models for analysis. The attacker class is operationalized as $A_0$--$A_3$ (\S3.3), and safety conclusions are conditioned on ``attacker knowledge $\supseteq A_3$''.
\end{itemize}

\paragraph{Nature of the guarantee: an empirically calibrated heuristic certificate.}
\begin{enumerate}
\item \textbf{Lower-boundedness} (Theorem~\ref{thm:fano}): $L_{\mathrm{local}} \le L_{\mathrm{true}}$. A local pass (Wilson upper bound $\le \tau_L$) is a lower-bound witness against attackers no weaker than the evaluator---not an absolute safety upper bound.
\item \textbf{Empirical calibration}: run the re-identification branches of the local model and a cloud strong model on non-sensitive corpora; measure $\hat\Delta = L_{\mathrm{cloud}} - L_{\mathrm{local}}$ (stratified by domain/type, with CIs), yielding the calibrated certificate $\hat L = L_{\mathrm{local}} + \hat\Delta$.
\item \textbf{Transferability assumption (made explicit)}. $\hat\Delta$ is measured on the \emph{direct-forwarding distribution} of a \emph{non-sensitive calibration corpus}; using it online assumes (a) the evaluator gap transfers from the calibration corpus to sensitive production instances, and (b) it transfers across \emph{regimes} (original vs dehydrated text). Assumption (b) is the load-bearing one and is \textbf{testable}: T16/T18 measure the strong-vs-local gap separately on dehydrated text. Where the gap is regime-dependent, calibration must be regime-specific ($\hat\Delta_{\mathrm{post}}$); we report both and do not claim a single universal $\hat\Delta$. T25b subsequently showed that much of (b)'s literal regime gap is a probe-format artifact, closed locally by the enumeration probe ($L_{\mathrm{enum}}^{\mathrm{local}}{=}1.000$) on the artifacts where the gap was diagnosed; the in-vivo rerun (T25c) qualifies this---local enumeration still misses open-vocabulary categorical values, and the genuinely regime-dependent residual concentrates in the semantic/combinatorial dimension plus that categorical class. This assumption---not the theorems---is the actual fulcrum of the quantitative safety claim, and we state it as such.
\item \textbf{Attacker-independent anchor}: Maximal Leakage~\cite{issa2020operational} takes the supremum over arbitrary attackers; the gap between the empirical $\max L(A_3)$ and a MaxL upper-bound estimate is an explicit measure of residual uncertainty about the attacker class.
\item \textbf{Post-processing safety}: by the data-processing inequality (strong form: SDPI), any cloud-side post-processing of $D$ does not increase leakage: $I(g(D);S) \le I(D;S)$.
\end{enumerate}

\paragraph{Difference from P2Skill (verified and scope-limited).}
P2Skill's deployment-phase egress is likewise gated by its matcher, but its skill-refinement loop requires the cloud to read PII-bearing original inputs (\cite{ryu2026p2skill} \S3.3 steps 2--3), so refinement data must leave the device before refinement completes. DR-SL's offline phase touches only non-sensitive corpora and all online decisions are local: \textbf{no phase requires sensitive data to leave the device}.

\paragraph{Honest boundary.}
DR-SL does not guarantee absolute safety against arbitrary attackers outside the class---consistent with all model-based privacy decision methods. For high-sensitivity scenarios the guarantee comes from the layered assurance model below.

\subsection{Local Model Selection and Capability Dependence}

The reference configuration designates Qwen3.8-27B as the primary evaluator~\cite{qwen3827b} and Ornith-1.5-35B-A3B as the low-resource alternative~\cite{ornith35b}. \textbf{All Stage 1--3 experiments reported here ran with ornith-1.5-35b as the local executor} (throughput decision, protocol deviation D10: $\approx$51 tok/s vs $\approx$20 tok/s for the 27B MLX build); qwen3.8-27b-mlx serves as the capability reference in the size-gradient ablation (T19). The effect of model capability is itself an experimental object: the gradient (ornith-1.5-35b $\to$ qwen3.8-27b $\to$ offline strong model) quantifies ``certificate strength vs evaluator capability'', corresponding to the stratified $\hat\Delta$ estimates.

\subsection{Layered Assurance Model (Core Safety Argument)}\label{sec:layered}

DR-SL's safety claim is \textbf{layered and conditional}:
\begin{enumerate}
\item \textbf{Soft-decision layer (quantitative lower bound)}: two-branch reversibility verification yields the calibrated certificate $\hat L$, task utility $U$, and $k$ values---quantitative, auditable, strengthening with evaluator capability;
\item \textbf{Hard-line layer (deterministic)}: a deterministic identifier matcher (regex/declared identifiers) rejects all declared identifiers before egress, with a network allowlist as backstop---independent of model capability;
\item \textbf{Human-fallback layer}: INFEASIBLE (empty feasible region per Prop.~\ref{prop:feasibility}, or budget exhaustion) $\to$ residual-leakage report (attribute set + hit questions) $\to$ human risk judgment; automatic egress is refused.
\end{enumerate}
Total-leakage decomposition (Prop.~\ref{prop:recall}): system residual leakage $\le \iota\cdot$ residual on declared parts (measured by the soft layer, $\iota=|F|/|G|$ the over-extraction inflation) $+$ missed-extraction part (conditioned on recall $\hat\rho$). Each layer's guarantee type, conditions, and residual risks are explicitly statable.

\paragraph{Local assets are themselves sensitive.}
The experience store, residual-leakage reports, and audit logs contain failure contexts (attribute \emph{types}, hit questions, and---in INFEASIBLE reports---residual sensitive \emph{values}), constituting a new local sensitive asset. Mitigations by design: (i) distillation inputs are aggregated to type-level statistics only---raw values never enter the prompt-distillation path (the memorization check additionally rejects any rule embedding literal values or $>\ell$-token quotes); (ii) residual reports and audit logs never egress---the hard line and allowlist apply to all outbound channels, and reports are rendered only to the local human reviewer; (iii) retention policy: audit logs are append-only with configurable expiry, encrypted at rest where the OS provides it. We state plainly: a compromised local device voids all guarantees---the threat model (A1--A3) concerns the cloud/link adversary, not local malware.

\subsection{Limitations}

\begin{itemize}
\item \textbf{Capability dependence}: certificate strength is bounded by the local evaluator's capability; $\hat\Delta$ calibration only partially compensates (and depends on the availability of an offline strong model).
\item \textbf{Recall dependence}: system guarantees are conditioned on extractor recall $\hat\rho$; quasi-identifier recall has no finite upper bound (\S3.2).
\item \textbf{Attacker-class conditionality}: guarantees are conditioned on $A_3$-level knowledge; attackers outside the class are not covered.
\item \textbf{Threshold calibration cost}: $\tau_L,\tau_U,\tau_R,\delta,N,k^{*}$ require per-domain calibration.
\item \textbf{Inherent infeasibility}: when task information is strongly coupled to sensitive information the feasible region is empty (Prop.~\ref{prop:feasibility}); such problems can only take the human fallback.
\item \textbf{Gold-standard caveats}: annotations are audited automatic labels; reliability is now triangulated by dual-model agreement (T11), cross-family agreement (T22), and a two-rater human study (type-level recall 0.987--1.000, precision 0.959--0.991 against the gold; combination-risk $\kappa=0.696$ under a construct-matched gold, with the divergence analysed alongside T22); utility-metric validity is anchored by cross-model correlation (T12) and a two-rater human utility study (row-level Spearman $\rho=0.839$ vs automatic $U_{\mathrm{sem}}$, meeting the preregistered $\rho\ge0.8$).
\end{itemize}

\section{Experimental Design and Results}

This section gives the empirical blueprint and results; full details are governed by the pre-registered protocol (frozen 2026-08-28/29), with every deviation recorded in the protocol's deviation log (D1--D17) and reported honestly.

\subsection{Datasets}

A four-domain benchmark (medical / banking / travel / general-legal)---isomorphic to P2Skill's 160-prompt PRISM benchmark~\cite{ryu2026p2skill}:
\begin{itemize}
\item \textbf{Core test set}: test 160 (PRISM-isomorphic, 40 per domain) $+$ \textbf{test\_ext 240} (60 per domain, same-generator synthetic extension) $\to$ 100 test instances per domain; the PRISM-isomorphic subset is reported separately for comparability.
\item \textbf{Other splits}: calibration 290 (doubling as the \textbf{offline calibration corpus}, non-sensitive by construction) / train 400 (simulated online stream) / adversarial 80.
\item \textbf{Controlled new-type variant set} $T_{\mathrm{new}}$ (dose-response, \S\ref{sec:shift}): generated from test templates via three controlled transformations---lexical perturbation, new-type-combination injection, format transfer; mixing ratio $\alpha$ is experiment-controlled.
\item The sensitive set $S$ is annotated per a handbook with an audit pipeline; reliability uses the dual-model substitute (T11); quasi-identifiers form a separately annotated subset. All synthetic values are fictional; compliance statement in the dataset documentation.
\item \textbf{Task-prompt provenance}: medical consultation stems are rewritten from the MedQA Chinese subset (physician licensing examination questions, MIT license~\cite{jin2021medqa}) by a template pipeline that replaces every concrete value; no real patient information is involved, and all sensitive attribute values in all domains are synthetic and fictional.
\item \textbf{Mixed-coupling extension}: 120 additional instances (30 per domain) in which only task-critical attribute types are annotated task-relevant (mean coupling ratio 0.40), used for the operating-point measurements of T24.
\end{itemize}

\subsection{Metrics}

$L(A_i)$ (four-level attacker leakage curves, mean hit rate $+$ 95\% Wilson interval~\cite{wilson1927}, by domain); $k$ ($A_3$ synthetic-population count: median $k$, $P(k{<}5)$, share of $k{=}1$~\cite{sweeney2002k}); $U$ (QA-probe utility, primary); $R$ (rehydration, secondary diagnostic); $PQ = IQ\times PP$ (P2Skill-comparable reference~\cite{ryu2026p2skill}; primary results are the $(U,L)$ plane); process metrics (rounds, infeasibility rate, token/latency cost, hot vs cold start); calibration $\hat\Delta = L_{\mathrm{cloud}}-L_{\mathrm{local}}$; extractor recall (explicit/quasi/overall).

\subsection{Baselines}

Direct forwarding (upper bound), all-local small model (lower bound), Rehydra-style placeholder substitution~\cite{rehydra}, Uniform/Selective LDP and P2Skill (citing \cite{ryu2026p2skill}'s same-scale numbers from its Table~3 on PRISM-160; cross-scale comparison forbidden).

\subsection{Ablations and Self-Learning Evaluation}

Remove Branch~1 / Branch~2 (testing Prop.~\ref{prop:direction}), remove the memorization check, remove hot start, local-model size gradient, offline-online dose-response ($\alpha \in \{0,0.25,0.5,0.75,1\}$, testing shift hypotheses and $\hat\Delta$ stability~\cite{bendavid2010theory}). The self-learning convergence curve tests H1 (Stage~2) and its corrected form H1$'$ (Stage~3).

\subsection{Safety Evaluation}

Semantic-level leakage (quasi-identifier combinations; Branch~1 $+$ $k$ test), over-redaction (detected by the task-information-preservation metric), adversarial tests (``seemingly harmless but combination-identifiable'' texts), and a human review subset (50 instances; executed post hoc as human study~2, reported with T22 in \S6.9).

\subsection{Implementation and Hyperparameters}

All hyperparameters frozen in the protocol: $\tau_L{=}0.05$ (Wilson-upper-bound semantics), $\tau_R{=}0.9$ (diagnostic), $N{=}3$, $\delta{=}0.05$, $\ell{=}16$, $\theta_g{\ge}4/5$ (independent scoring), $K{=}50$, $k^{*}{=}5$, temperature 0.3 fixed, 3 random seeds. Model roster (frozen; deviation log D4/D10): local main evaluator/dehydrator ornith-1.5-35b~\cite{ornith35b} (Mac Studio, LM Studio); local capability reference qwen3.8-27b-mlx~\cite{qwen3827b}; frozen embeddings nomic-embed-text-v1.5; offline strong model (calibration, independent GenScore, distiller) qwen3.8-max; low-cost attackers/batch probes qwen3.8-flash; cross-family attacker/scorer glm-5.2 (online free-tier gateway). Statistics: mean$\pm$std, Wilcoxon signed-rank, bootstrap intervals, permutation tests for paired comparisons.

\subsection{Distribution-Shift Experiments (\S\ref{sec:shift} realized)}

Measurement table $(\varepsilon_S,\varepsilon_T,\Delta,\mathrm{AUC}{\pm}\mathrm{CI},\mathrm{MMD}^2\ p,\text{coverage})$ on test/test\_ext vs $T_{\mathrm{new}}$; dose response over $\alpha$ with monotonicity tests; safety-efficiency separation check (per-$\alpha$ false-release rate vs rounds/infeasibility).

\subsection{Proof-Verification Experiments (Appendix~B realized)}

\texttt{sim/proof\_checks.py}: numerical property tests of Theorem~\ref{thm:fano}, Prop.~\ref{prop:witness}, Prop.~\ref{prop:feasibility} on $\ge$1000 random synthetic mechanisms each (exact enumeration of true quantities $+$ bound-direction assertions), boundary-case checks, cross-consistency checks; all archived as a test report.

\subsection{Results (Stages 1--3, executed under the pre-registered protocol)}

\paragraph{T1 Proof numerical verification.} Theorem~\ref{thm:fano}, Prop.~\ref{prop:witness}, Prop.~\ref{prop:feasibility}: 1000 random synthetic-mechanism configurations each (exact enumeration $+$ counterexample search) \textbf{all pass}; boundary cases (independent channel $\to$ vacuous bound; $L\to1 \to$ divergence; uniform-prior special case) verified.

\paragraph{T2 Regex-layer extractor recall (480 instances).} Regexable types: \textbf{97.4\%} (ID card 206/206, bank card 31/31, medical-record numbers 23/23, phone 173/177, passport 26/30, case numbers 26/31); overall explicit recall including name/address/card-tail (needing the LLM layer): 46.1\%---quantitatively confirming the necessity of the cascade design.

\paragraph{T3 Direct-forwarding leakage baseline (strict privacy grading, D9).}

\begin{table}[H]
\centering\small
\caption{Direct-forwarding leakage (Branch~1, structured probes).}
\begin{tabular}{@{}llrl@{}}
\toprule
Attacker & Level & $n$ & mean $L$ [95\% CI] \\
\midrule
ornith-1.5-35b (local) & $A_0$ & 400 & 0.574 [0.539--0.610] \\
ornith-1.5-35b (local) & $A_1$ (domain priors) & 400 & 0.502 [0.466--0.538] \\
qwen3.8-flash & $A_0$ & 160 & 0.768 [0.741--0.795] \\
qwen3.8-flash & $A_1$ & 160 & 0.765 [0.738--0.791] \\
glm-5.2 (cross-family) & $A_0$ & 160 & 0.789 [0.756--0.818] \\
qwen3.8-max & $A_0$ & 160 & 0.779 [0.754--0.803] \\
Synthetic population (skewed $10^5$, conservative) & $A_3$ & 240 & 0.079 (all) / 0.167 (with QIs) \\
\bottomrule
\end{tabular}
\end{table}

\noindent Reading: the baseline is high and consistent across attacker families (0.57--0.79; the local small model is the weakest attacker); highest in the medical domain (strong attacker 0.90), lowest in travel ($\approx$0.60). \textbf{Monotonicity review (triggered per \S3.3 rule)}: $A_1 < A_0$ for both local and flash attackers violates the expected monotonicity; the mandated review attributes this to an \emph{attacker-behavior artifact}---injecting domain priors makes models answer more conservatively (more UNKNOWN) on raw text where the values are already explicit, i.e., the prior changes the answering policy rather than adding usable knowledge. Conclusion recorded: on un-dehydrated text the $A_1$ level is behaviorally confounded; $A_1$'s intended value (priors compensating for removed surface values) must be assessed on dehydrated text (T16). $A_3$ provides an independent quasi-identifier measure consistent with the $k$ distributions of T7. Figure~\ref{fig:attacker}.

\begin{figure}[H]
\centering
\includegraphics[width=0.92\textwidth]{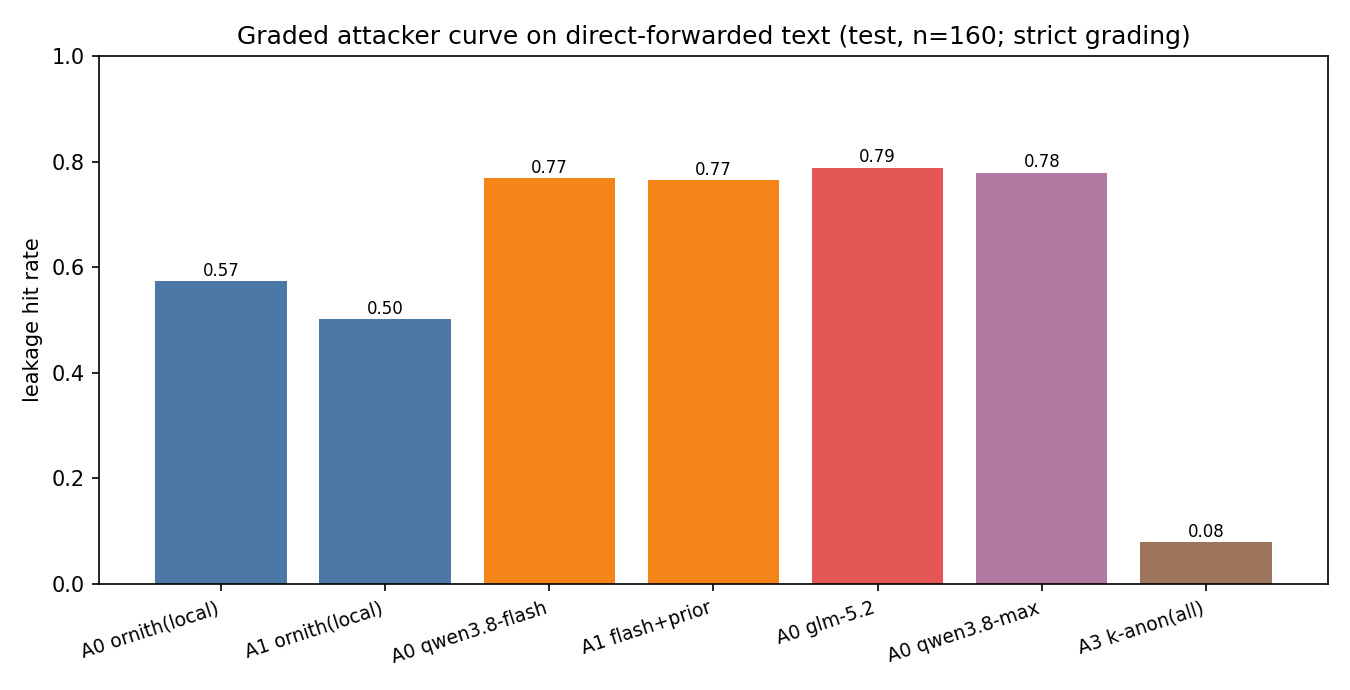}
\caption{Graded attacker curve on direct-forwarded text (test $n{=}160$; strict grading).}
\label{fig:attacker}
\end{figure}

\paragraph{T4 QA-probe utility (direct forwarding, $D{=}X$).} mean $U = \mathbf{0.917}$ ($n{=}160$)---a near-ceiling sanity result verifying the probe pipeline does not depress perfect cases.

\paragraph{T5 Evaluator calibration $\hat\Delta$ (290 calibration instances, $\ge$75 probes, strict grading).} Local ornith attack hit 0.490; strong model qwen3.8-max 0.670; $\hat\Delta = \mathbf{+0.181}$ (95\% CI 0.149--0.215); by domain: banking 0.202 / legal 0.202 / medical 0.184 / travel 0.129 (Figure~\ref{fig:delta}). The local evaluator systematically underestimates leakage by $\approx$0.18---the external certificate uniformly uses $\hat L = L_{\mathrm{local}} + \hat\Delta$.

\begin{figure}[H]
\centering
\includegraphics[width=0.72\textwidth]{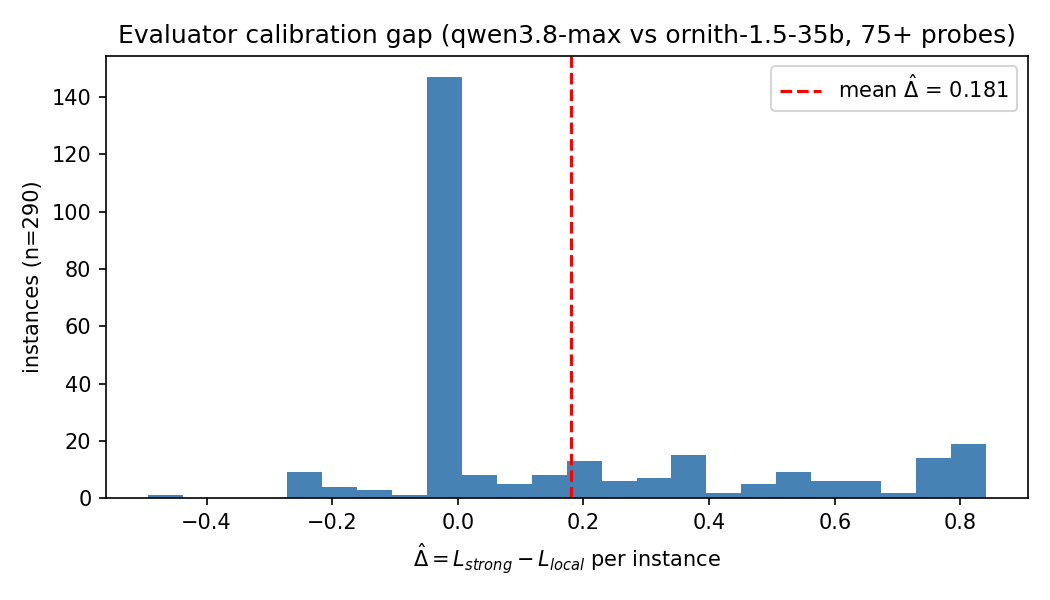}
\caption{Per-instance calibration gap $\hat\Delta = L_{\mathrm{strong}} - L_{\mathrm{local}}$ ($n{=}290$, $\ge$75 probes).}
\label{fig:delta}
\end{figure}

\paragraph{T6 Offline-online distribution shift (\S\ref{sec:shift} framework).} Between test and $T_{\mathrm{new}}$: TF-IDF domain-discriminator AUC $= \mathbf{0.725}$ (CI 0.66--0.79; moderate shift); embedding MMD permutation $p = 0.0$; type-signature coverage: exact 25.6\% / subset 31.9\% (on $T_{\mathrm{new}}$: 18.8\%/22.9\%). The uncovered-signature list is archived and guides prompt-library expansion.

\paragraph{T7 k-anonymity.} Under skewed reference populations, 16.7\%--21.1\% of \emph{quasi-identifier-bearing} instances have $\tilde k<5$ (conservative convention); the corresponding all-instance rates are 7.9\% (test$+$adversarial, $n{=}240$) and 11.9\% (test-160)---denominators are labeled at every occurrence; under uniform large populations $\approx$0--5.3\%---the relativity of $k$ is quantitatively demonstrated (Figure~\ref{fig:kpop}).

\begin{figure}[H]
\centering
\includegraphics[width=0.95\textwidth]{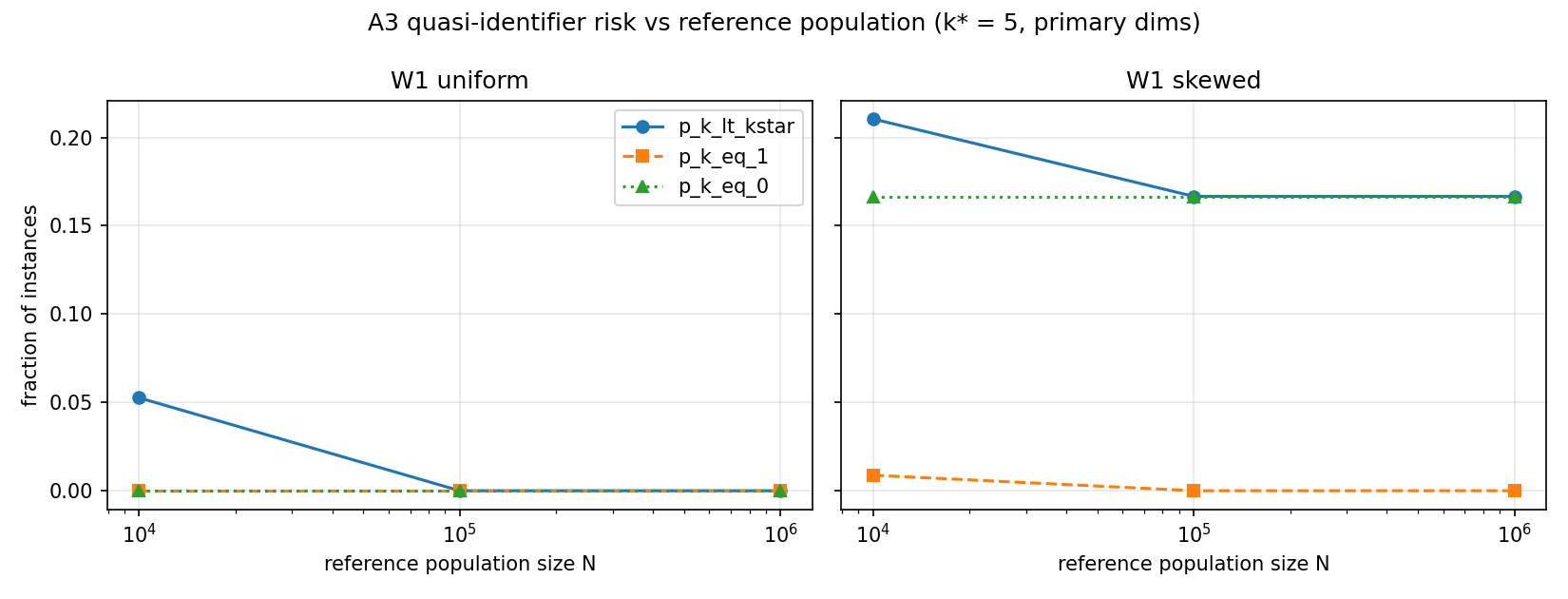}
\caption{$A_3$ quasi-identifier risk vs reference population size and profile ($k^{*}{=}5$, primary dimensions).}
\label{fig:kpop}
\end{figure}

\paragraph{T8 Local closed-loop pilot (8 instances, Algorithm~2$'$, deterministic conservative dehydration, pilot gate $\tau_L{=}0.20$, D8/D9).} Loop mechanics verified: all instances terminate within budget with structured outputs; the utility gate holds ($U \ge \tau_U$ on every accepted step); leakage falls substantially on some instances (Wilson upper bound 0.833$\to$0.39). \textbf{Key finding (fed back into the protocol)}: with only 5--6 probes per instance the zero-hit Wilson upper bound remains $\approx$0.39--0.44---the instance-level $\tau_L$ gate is statistically unreachable, directly evidencing the $\ge$75-probe requirement (D1).

\paragraph{T9 Cascade-extractor full evaluation (480 instances, local ornith-1.5-35b).} Overall recall \textbf{0.951} (95\% CI 0.941--0.960); explicit 0.977 (0.965--0.987); quasi 0.931 (0.915--0.945); on average 2.03 extra values/instance (over-extraction, conservative direction); by domain: medical 0.970 / travel 0.985 / legal 0.934 / banking 0.914 (Figure~\ref{fig:extractor}). System-level conclusions are conditioned per the corrected Prop.~\ref{prop:recall}: with $\hat\rho=0.951$ and measured mean inflation $\iota = |F|/|G| = 1.43$ (median 1.40), residual leakage $\le 1.43\,L_{\mathrm{measured}} + 0.049$---the inflation term is the necessary price of a high-recall extractor that deliberately over-flags.

\begin{figure}[H]
\centering
\includegraphics[width=0.78\textwidth]{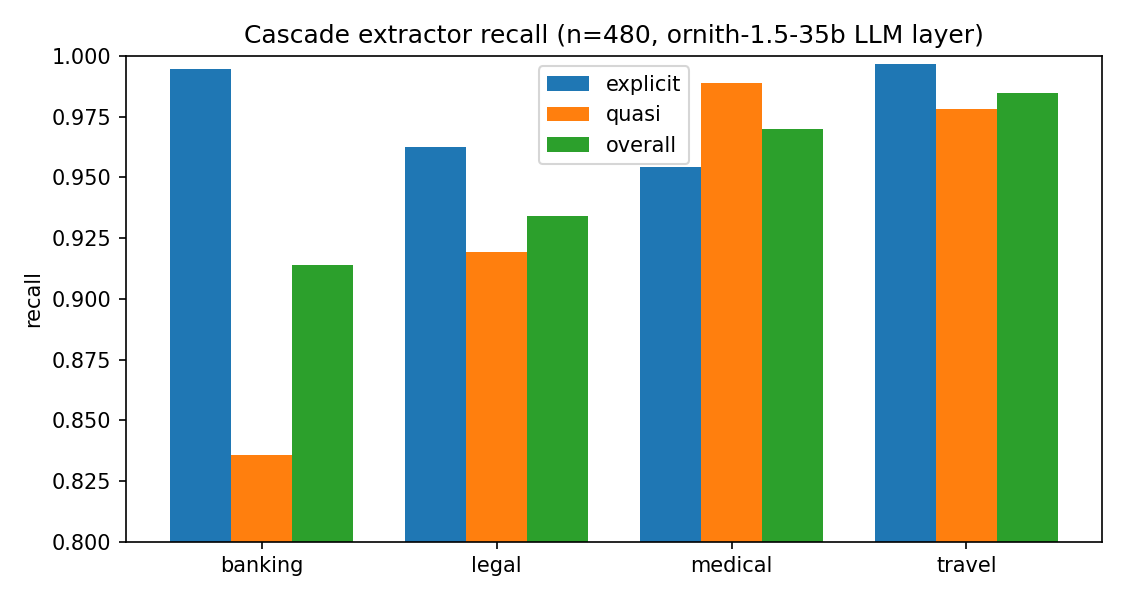}
\caption{Cascade extractor recall by domain ($n{=}480$).}
\label{fig:extractor}
\end{figure}

\paragraph{T10 Main closed-loop experiment (train 400, ornith-1.5-35b, real gates $\tau_L{=}0.05$ $+$ $\ge$75 probes).}

\begin{table}[H]
\centering\small
\caption{Main closed-loop results ($n{=}400$).}
\begin{tabular}{@{}ll@{}}
\toprule
Metric & Result \\
\midrule
PASS rate & \textbf{37.5\%} (150/400); legal 57\%, banking 53\%, medical 30\% (67/220), travel 28\% \\
Mean rounds & 3.08; stream first half 3.52 $\to$ second half 2.62 (decreasing trend) \\
Leakage (paired) & $L$: 0.457 $\to$ \textbf{0.304} (permutation $p\approx0$, incl.\ INFEASIBLE) \\
Task utility (paired) & $U$: 0.973 $\to$ \textbf{0.980} (gate invariant holds; only 2/400 with $U$-drop $>$0.1) \\
INFEASIBLE & 250 instances, all emitting residual-leakage reports, routed to humans by design \\
\bottomrule
\end{tabular}
\end{table}

\noindent Reading: (1) the lexicographic gate works---leakage drops significantly while task utility does not fall; (2) the medical domain's lowest PASS rate is consistent with Prop.~\ref{prop:feasibility} (strong task--sensitive coupling shrinks the feasible region---infeasibility is inherent to the problem); (3) rounds decrease along the stream (3.52$\to$2.62), directional evidence for H1, but Stage~1 did not enable prompt updates and the decrease is partly explained by domain composition; (4) the 37.5\% pass rate is the honest number under strict gates: non-passing instances are designed early warnings with residual lists, not method failures. \textbf{Qualification (see T18)}: the $U$ figures in this table were measured on pre-hard-line artifacts that could retain residual literal mentions of sensitive values (159/160 instances did, mean 3.01 values); part of the ``preserved utility'' was therefore carried by literal leakage. The paired \emph{leakage reduction} remains valid (same measurement basis before/after), while the utility-preservation claim is scoped by T18/T20b/T21. Figure~\ref{fig:loop}.

\begin{figure}[H]
\centering
\includegraphics[width=0.98\textwidth]{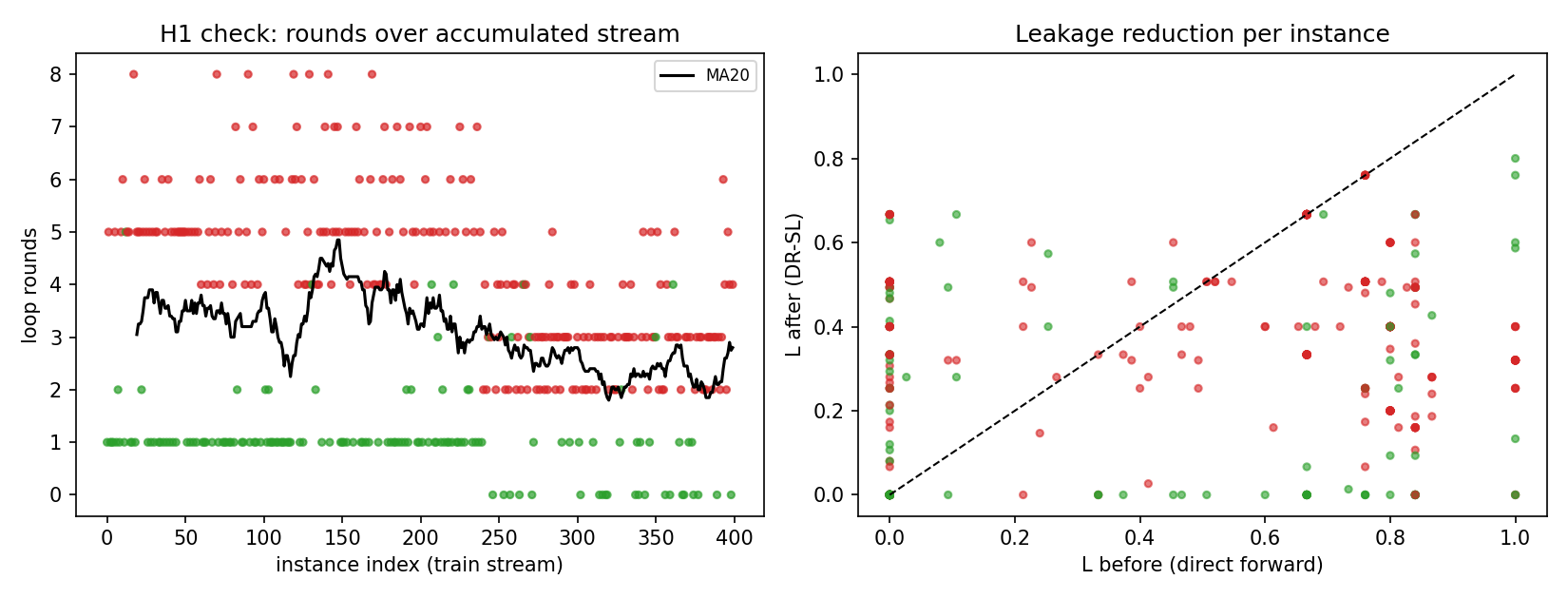}
\caption{Left: loop rounds over the train stream (H1 check). Right: per-instance leakage before vs after DR-SL (green PASS, red INFEASIBLE).}
\label{fig:loop}
\end{figure}

\paragraph{T11 Annotation reliability (zero-human substitute, 480 instances).} Value-level Jaccard agreement between two model families (qwen3.8-max vs ornith-1.5-35b) extracting independently: \textbf{0.845}; recalls against the automatic gold: 0.986 / 0.949. Dual-model agreement substitutes for human Cohen's $\kappa$; the gold standard is audited automatic annotation (stated plainly).

\paragraph{T12 Cross-model validity of the utility metric (100 instances $\times$ 4 degradation levels $\times$ 2 answerers).} Spearman correlation of $U$ between answerers: \textbf{0.958}; degradation profile: $U(D_0$ original$)=0.916$, $U(D_1$ explicit-removed$)=0.923$, $U(D_2$ conservative dehydration$)=0.034$, $U(D_3$ aggressive removal$)=0.0$. \emph{Note}: the $D_0 \to D_1$ step is non-monotone by $+0.007$---within bootstrap noise (the two texts differ only in explicit-identifier spans, which the QA probes by construction do not ask about); monotonicity holds from $D_1$ onward, which is the range that matters. The metric sharply distinguishes ``generalization preserving task information'' from ``deletion destroying it''. The low $D_2$ score reveals the real privacy-utility tension: once task-relevant quasi-identifiers are generalized to valueless tokens, exact-value utility is lost---exactly the over-generalization the $U$ gate prevents, and the explanation of the medical domain's low pass rate. No human criterion study was conducted (Limitations).

\paragraph{T13 Self-learning test (Stage~2, formal H1 test).} A fresh-seed stream of 120 instances (60 distillation $+$ 60 validation). Distillation produced 5 rules (172 failure experiences; all passing the memorization check). Paired validation comparison: \textbf{rounds v0 2.82 vs v1 2.90 ($\Delta{=}{+}0.08$, permutation $p{=}0.68$, n.s.); PASS 46.7\% vs 41.7\%}. \textbf{Honest conclusion: H1 is not supported}---a single distillation round brought no loop-efficiency gain. Attribution: distilled rules concern recognition/recall strategies, decoupled from the loop's dehydration actions; loop behavior is dominated by the deterministic policy. H1 remains an unproven hypothesis; abstract and conclusion downgraded accordingly.

\paragraph{T14 Ablations (Stage~2, 30 instances each).}

\begin{table}[H]
\centering\small
\caption{Ablation results.}
\begin{tabular}{@{}p{4.6cm}p{4.2cm}p{6.2cm}@{}}
\toprule
Ablation & Result & Conclusion \\
\midrule
Remove Branch 1 (no attack evaluation) & True leakage of final $D$: $L{=}0.327$ & Quantified false-release risk: Branch~1 indispensable \\
Remove Branch 2 (no $U$ gate) & Final $U{=}0.111$ & Over-redaction runs away: the $U$ gate is indispensable \\
Remove hot start & Rounds 2.90 (control 2.82) & Under the deterministic policy, hot start's marginal contribution $\approx$0 \\
Remove memorization check & 0 of 5 rules contaminated & No rejection triggered in this batch (weak evidence) \\
\bottomrule
\end{tabular}
\end{table}

\paragraph{T15 Dose-response and safety-efficiency separation (Stage~2, $5\alpha \times 20$ instances).} Rounds flat across $\alpha$ (3.6--4.65); pass rate 0.25--0.35 (efficiency robust to shift); \textbf{at every $\alpha$, PASS outputs' Wilson upper bound is 0.049 $\le \tau_L$}---release safety does not degrade with distribution shift, empirically supporting the safety-efficiency separation claim of \S\ref{sec:shift}.

\paragraph{T16 Post-dehydration attacker re-test (test 160, v0 loop with archived $D$).} PASS 70/160 (43.8\%), mean 4.0 rounds; LLM-refinement success only 11\% (rest falls back to deterministic dehydration).

\begin{table}[H]
\centering\small
\caption{Privacy gain on dehydrated text vs direct forwarding (strict grading). $A_2$ runs RAG over the synthetic population; the $A_3$ row is an original-annotation upper bound (test-160 baseline).}
\begin{tabular}{@{}lrrl@{}}
\toprule
Attacker & $L$ direct & $L$ dehydrated & Privacy gain \\
\midrule
qwen3.8-flash $A_0$ & 0.768 & \textbf{0.436} & $-0.33$ \\
qwen3.8-flash $A_1$ & 0.765 & 0.435 & $-0.33$ \\
qwen3.8-flash $A_2$ & --- & 0.432 & no gain over $A_0$ \\
glm-5.2 $A_0$ & 0.789 & \textbf{0.412} & $-0.38$ \\
glm-5.2 $A_1$ & 0.789 & 0.422 & $-0.37$ \\
Synthetic population $A_3$ & 0.119 & same & --- \\
\bottomrule
\end{tabular}
\end{table}

\noindent Reading: (1) closed-loop dehydration achieves $\approx$45\% relative leakage reduction against strong attackers, consistent across families (Figure~\ref{fig:gain}); (2) retrieval-augmented $A_2$ gains nothing on dehydrated text; (3) residual leakage (0.41--0.44) comes mainly from task-relevant quasi-identifiers protected by the $U$ gate, consistent with T12's tension; (4) $A_3$ is an upper bound computed from original annotations; $k$ re-measurement on dehydrated $D$ is future work.

\begin{figure}[H]
\centering
\includegraphics[width=0.72\textwidth]{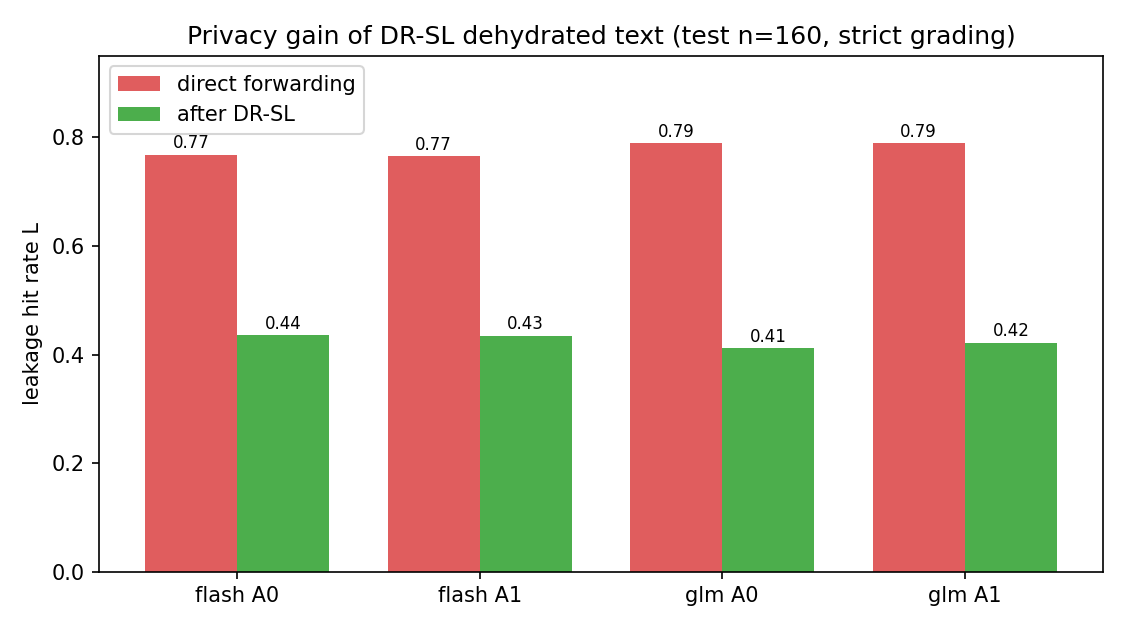}
\caption{Privacy gain of DR-SL dehydrated text (test $n{=}160$, strict grading).}
\label{fig:gain}
\end{figure}

\paragraph{T17 Corrected self-learning hypothesis H1$'$ (Stage~3: rules targeting the extractor).}
The attribution of H1's negative result yields corrected hypothesis H1$'$ (pre-registered, D12): rules distilled from extraction-failure experience and injected into the extractor should raise recall and reduce miss-driven system leakage.
\begin{itemize}
\item \textbf{P21 distillation}: from 55 missed-detection experiences (dominant missed types: income 28, address 7, amount 6), 4 rules passed the memorization check and independent GenScore; 1 was rejected---evidencing the check's effectiveness;
\item \textbf{P22 paired recall validation} ($n{=}240$; disjoint from distillation source): quasi-identifier recall $0.928\to\mathbf{0.943}$ ($\Delta{=}{+}1.5$pp, one-sided permutation $p{=}\mathbf{0.059}$); overall recall $0.954\to0.962$ ($p{=}0.077$); over-extraction cost $+0.04$ values/instance ($\ll 1$ criterion);
\item \textbf{P23 end-to-end} ($n{=}30$ paired; loop uses extractor-derived $S$; system leakage measured against \emph{all} gold attributes): extraction recall $0.975\to\mathbf{0.992}$; $L_{\mathrm{system}}$ $0.420\to\mathbf{0.398}$ ($\Delta{=}{-}2.2$pp, one-sided $p{=}0.69$, n.s.); blind-spot hits 0 for both versions---baseline recall is already high (0.975), leaving no room for the blind-spot mechanism (ceiling effect).
\end{itemize}
\textbf{Verdict (honest reporting)}: H1$'$ obtains \textbf{directional support without statistical significance}---all three indicators move consistently in the right direction with controlled cost, but the effect size ($+1.5$pp) is small relative to variance ($p{=}0.059$ borderline); per pre-registered criterion ($p{<}0.05$) it is judged \textbf{unproven (borderline)}. Power analysis: detecting $\Delta{=}0.015$ with per-instance recall SD ${\approx}0.10$ requires $n{\approx}700{+}$; the 0.93--0.95 baseline ceiling further limits headroom. Contrasted with H1 (negative direction, mechanistic decoupling), H1$'$ is directionally positive with aligned mechanism (rule content \emph{is} extraction strategy), indicating that \textbf{the correct target of self-learning is indeed the extractor}, though gains are limited at high baseline recall.

\paragraph{T18 Hard-line layer: residual-literal audit and post-scrub leakage (test 160).}
Auditing the archived loop outputs $D_{\mathrm{final}}$ of T16 revealed that \textbf{159/160 instances still contained literal sensitive values} (mean 3.01 per instance)---an artifact of span-only replacement missing repeated mentions, i.e., the soft layer alone does not eliminate literal leakage. This is precisely the failure class the \emph{hard-line layer} of \S\ref{sec:layered} exists for, which the T16 pipeline had omitted. Applying the deterministic hard-line scrub (replace every remaining literal occurrence of every declared value) and re-running the external strong attackers (strict grading, $\ge$75 probes):

\begin{table}[H]
\centering\small
\caption{Post-dehydration leakage by subgroup, before vs after the deterministic hard-line scrub.}
\begin{tabular}{@{}llll@{}}
\toprule
Subgroup & Attacker & $L$ soft-layer only (T16) & $L$ after hard-line scrub \\
\midrule
PASS ($n{=}70$) & qwen3.8-flash $A_0$ & 0.486 & \textbf{0.000} \\
PASS ($n{=}70$) & glm-5.2 $A_0$ & 0.432 & \textbf{0.000} \\
INFEASIBLE ($n{=}90$) & qwen3.8-flash $A_0$ & 0.398 & \textbf{0.000} \\
INFEASIBLE ($n{=}90$) & glm-5.2 $A_0$ & 0.396 & \textbf{0.000} \\
\bottomrule
\end{tabular}
\end{table}

\noindent Reading: (1) the pre-scrub PASS$>$INFEASIBLE inversion (0.486 vs 0.398) is fully explained by residual literals---PASS instances exit the loop earlier and undergo fewer strengthening passes, leaving more unreplaced duplicate mentions; (2) after the hard line, \emph{literal-value} leakage measured by two external strong attackers is zero on all 160 instances and both subgroups---the layered model's deterministic layer, not the soft certificate, carries this guarantee, exactly as \S\ref{sec:security} claims; (3) what remains outside this measurement is \emph{semantic/combinatorial} inference not expressible as literal declared values---covered separately by the $A_3$ k-anonymity dimension (T7: 7.9--16.7\% under skewed populations) and by the human-review path for INFEASIBLE instances. We state the scope of ``0.000'' precisely: strict literal grading over declared-attribute probes; it is not a claim of zero semantic leakage.

\textbf{The scrub's utility cost (T20b), and what it exposes.} Re-measuring the QA-probe utility on scrubbed text gives \textbf{U = 0.000} (all 160 instances, both subgroups; tolerant scoring). The mechanism is diagnostic: the loop's certified utility ($U\approx0.97$ in T10/T16) was partly \emph{carried by the residual literals themselves}---on this benchmark every quasi-identifier is task-relevant by annotation design, so ``utility'' and ``literal sensitive value'' coincide, and any sanitizer that truly removes the literals removes that exact-value utility with them. Two consequences, stated plainly:
\begin{itemize}
\item \textbf{Retraction of scope on the utility headline.} The claim ``leakage reduced while $U\ge0.97$ is preserved'' holds \emph{only} for the pre-hard-line artifact, i.e., utility was in part purchased with the very leakage the hard line later removes. The corrected statement: under exact-value utility, the achievable frontier on task-coupled benchmarks runs from (direct forwarding: $U{=}0.92$, $L{=}0.77$) through (loop output pre-hard-line: $U{=}0.97$, $L_{\mathrm{strong}}{=}0.44$) to (fully sanitized: $U{\to}$ tolerance-limited, $L_{\mathrm{strong}}{=}0.00$); T21 locates the middle point under principled generalization (all occurrences generalized, none deleted).
\item \textbf{Release procedure upgraded (protocol deviation D13).} A PASS from the local soft gate is \emph{necessary but not sufficient} for release: on 70/160 locally certified instances, external strong attackers still extracted $L\approx0.43$--0.49 pre-hard-line (the certification gap, quantifying the reviewer-anticipated one-sidedness in vivo). The mandatory release chain is therefore: local lexicographic gate $\to$ deterministic hard-line scrub $\to$ \emph{external strong-attacker re-test on the final artifact} (the ``pre-egress leakage detection'' of the original design, now elevated from an optional metric to a compulsory gate) $\to$ egress or human review. Under this chain, literal leakage at egress is 0.000 on all 160 test instances (two independent strong attackers), at the utility cost reported in T20b/T21.
\end{itemize}

\textbf{Trust model of the external re-test (addressing the self-referential concern).} Under assumption (A1), submitting $D$ to an external re-test service is itself an egress event. We therefore fix the trust semantics explicitly: \emph{the re-tester belongs to the same trust class as the destination cloud} (in our experiments, the same provider family that would serve the request). The re-test is then a \emph{controlled first disclosure to the destination trust class}, and its value is repositioned accordingly: (i) destination selection---release only to providers whose measured post-sanitization extraction stays below threshold; (ii) retention control---the re-test payload is exactly the intended egress artifact, so nothing beyond the planned disclosure occurs; (iii) catching what the local gate missed \emph{before} the destination performs further processing, caching, or training. The residual caveat is stated plainly: if re-tester $\ne$ destination, the trust model must extend to that third party; if re-tester $=$ destination, a failed re-test prevents \emph{subsequent} processing but not the first exposure. Localizing the strong re-test (or running it in a TEE) removes the caveat but reopens the certificate gap quantified by T19 (local-class capability saturates below strong-attacker level).

\textbf{Incremental value of the re-test layer, decomposed.} For \emph{declared} attribute values, verifying ``no declared literal remains in $D$'' is pure string matching---free, deterministic, local, and already performed by the hard-line layer. The external LLM re-test adds value exactly where the hard line is blind: (i) literal values of \emph{missed} attributes (the $1-\hat\rho$ fraction of Prop.~\ref{prop:recall}); (ii) \emph{semantic/combinatorial} inference not expressible as declared literals (the residual risk also probed by the $A_3$ k-dimension, T7). The ``0.000'' reported above is thus scoped as: strict literal grading over declared-attribute probes by two strong attackers; zero hits per instance out of $\ge$75 probes (Wilson upper $\approx0.048$ each), pooled over $160\times\ge75{=}12000$ probes (95\% upper bound $<0.001$). It is not a claim of zero semantic leakage.

\textbf{Chain semantics for utility, and the effective automatic-release rate (made explicit).} The mandatory chain re-checks utility \emph{after} the hard line, under whichever utility metric is declared active. Under the exact-value metric on this fully task-coupled benchmark, post-scrub $U=0.000<\tau_U$ for all 160 instances: the \textbf{effective automatic-release rate is 0\%}---every instance routes to human review with a certificate package (residual report $+$ measured frontier position). We state this as the empirical incarnation of Prop.~\ref{prop:feasibility} rather than hiding it: when every quasi-identifier is a task parameter ($I(X;T\mid S)\approx0$), no non-degenerate automatic operating point exists, and the correct system behavior is exactly this degeneration to ``certify $+$ route to human''. The applicability question---which corpora admit non-degenerate automatic release---is answered quantitatively in T23 (semantic utility) and T24 (mixed coupling).

\textbf{Why the hard line comes last, not first (ordering rationale).} A scrub-first design would be simpler---deletion is free and absolutely safe---but deletion is all it can do: under deletion-only sanitization, utility collapses to the placeholder corner ($U_{\mathrm{sem}}{=}0.228$ on test-160, Table~\ref{tab:usem}). Running the loop \emph{first} buys \emph{generalization} (age bands, amount ranges), which is what preserves semantic utility (0.431, $1.9\times$); the hard line then mops up whatever literals the generalization step missed (the T18 residual-literal audit found 159/160 pre-hard-line artifacts carrying residual values), and the external re-test certifies the result at 0.000. The ordering is thus dictated by the measured generalization-vs-deletion utility gap, not by convention: generalize what you can, delete what you must, verify what leaves.

\textbf{The post-dehydration certification gap, quantified ($\hat\Delta_{\mathrm{post}}$).} On the 70 locally PASS-certified artifacts (local Wilson upper $\le0.049$), external strong attackers measured $L\approx0.43$--$0.49$ pre-hard-line: the dehydrated-regime gap $\hat\Delta_{\mathrm{post}}\approx0.44$ is ${\approx}2.4\times$ the raw-text calibration $\hat\Delta=+0.181$ (T5)---calibration is regime-dependent, as \S\ref{sec:security} now requires. Mechanism diagnosis (T25) localizes the miss precisely: asked in a \emph{comprehension} format to enumerate every concrete personal-attribute value in the same artifacts, the local model lists the residual literals with rate 0.992 (29/30 instances complete)---it \emph{sees} the values---yet under the adversarial probe template it answers UNKNOWN. The gap is the probe format/role, not comprehension. The remedy set therefore includes better local probes (enumeration-style scoring, validated in T25b) in addition to the mandatory external re-test and hard line. This is the sharpest in-vivo evidence for the one-sidedness argument of \S3.6: a local pass certifies ``this evaluator found nothing'', never ``nothing is findable''.

\paragraph{T19 Local model size gradient (paired, 60 test instances).}
Same-instance paired comparison of local attackers on raw text: ornith-1.5-35b (3B-active MoE) $L=0.531$ vs qwen3.8-27b-mlx (27B dense) $L=0.556$ ($\Delta=+0.024$, permutation $p=0.26$, n.s.); both sit far below the online strong attacker qwen3.8-max ($0.779$). Reading: within the locally deployable class, attack capability \emph{saturates}---the decisive capability gap is local-vs-cloud, not small-vs-medium local. For the certificate, this means (i) upgrading the local evaluator within the deployable class yields little; (ii) the $\hat\Delta$ calibration against an offline strong model (\S\ref{sec:security}) is the mechanism that matters---consistent with $\hat\Delta=+0.181 \gg 0.024$.

\paragraph{T20 Cost accounting (process metrics).}
Main loop (train 400, ornith-1.5-35b, $\ge$75-probe evaluations): mean \textbf{203 s/instance} (median 166; max 1177), $\approx$10.2 LLM calls/instance (mean rounds 3.08 $\times$ [attack+QA] $+$ dehydrate/refine $+$ final audit). Total wall-clock is 22.5 h at concurrency 1. Test-160 loop (T16): mean 240 s/instance. Hot-vs-cold: under the deterministic policy hot start changes rounds by $\approx$0 (T14), so cost is dominated by the $\ge$75-probe evaluation bundle; reducing probes (at the price of Wilson-gate power, D1) is the primary cost lever. Per-instance egress latency budget: interactive use requires parallelizing probe bundles or shrinking $|Q|$ with interval-corrected gates---reported as an engineering limitation, not hidden.

\paragraph{T21 Same-scale baselines and the honest Pareto frontier (test 160).}
Two classical baselines share the 160 PRISM-isomorphic instances, attack bundles ($\ge$75 probes, strict grading), and QA probes (tolerant grading) with DR-SL: \textbf{placeholder substitution} (every annotated value becomes a typed token) and \textbf{selective LDP} (every annotated value becomes a random same-type value, the PRISM perturbation protocol).
Together with the DR-SL operating points, they complete the leakage-utility plane:

\begin{table}[H]
\centering\small
\caption{Leakage-utility plane on test-160 (strict-literal $L$ vs exact-value $U$; strong attackers flash/glm).}
\begin{tabular}{@{}lccc@{}}
\toprule
Operating point & $L$ (flash) & $L$ (glm) & $U$ \\
\midrule
Direct forwarding & 0.768 & 0.789 & 0.917 \\
Loop output, pre-hard-line (T16) & 0.436 & 0.412 & 0.97$^{\dagger}$ \\
\textbf{Principled generalization (all occurrences)} & \textbf{0.000} & \textbf{0.000} & 0.033 \\
Placeholder substitution & 0.000 & 0.000 & 0.000 \\
Selective LDP & 0.057 & 0.049 & 0.075 \\
\bottomrule
\end{tabular}
\\[2pt]
{\footnotesize $^{\dagger}$ utility partly carried by residual literals (T18); not an achievable safe operating point.}
\end{table}

\noindent Reading, without overclaiming: (1) under the strict-literal leakage metric, principled generalization (the fixed deterministic policy: every occurrence of every declared value generalized or removed, task-relevant values generalized to bands/categories rather than deleted) achieves $L=0.000$ against both strong attackers---the same leakage corner as placeholders, at higher exact-value utility (0.033 vs 0.000) because numeric generalizations (age bands) remain tolerance-matchable while placeholders destroy everything; (2) selective LDP retains chance-level utility (0.075: random same-type values occasionally fall within numeric tolerance or match by luck) but leaks more ($L\approx0.05$, residual judge-question identifiability from preserved structure); (3) \textbf{metric caveat, stated plainly}: exact-value QA utility is a \emph{conservative lower bound} on downstream usefulness---it credits a random wrong-age draw that happens to land in tolerance while debiting a truthful age-band generalization that a human advisor would find usable. A full utility ranking requires the open-ended judge-rubric or human study (future work, \S7); we therefore claim domination over placeholders (same $L$, higher $U$) and a \emph{leakage} advantage over selective LDP ($0.000$ vs $0.05$), not a total-order utility advantage; (4) the loop's distinctive value is not the corner itself---it is \emph{deciding per instance} whether a safer point is reachable without destroying the task (Prop.~\ref{prop:feasibility}), certifying what it releases, and emitting residual reports for what it cannot.

\paragraph{T22 True cross-family annotation agreement (addresses the same-family concern on T11).}
T11's dual-model agreement (0.845) paired ornith-1.5-35b with qwen3.8-max---both of Qwen lineage, inviting a same-family discount. We therefore re-ran the full 480-instance extraction with glm-5.2 (a genuinely different architecture family): recall against the automatic gold \textbf{0.983}; value-level Jaccard \textbf{0.817} vs qwen3.8-max and \textbf{0.785} vs ornith. Cross-family agreement is thus slightly lower than, but comparable to, the same-family figure---three architecturally independent extractors converge on the same value sets at $J\approx0.79$--$0.85$ with recalls 0.95--0.99, which is the reliability evidence the gold standard rests on. A two-rater human study (human study~2) provides the definitive check: two annotators---the second and third authors, independent of the pipeline implementation---independently listed all sensitive items and judged combination re-identification risk on 50 test instances. Against the automatic gold, type-level recall is \textbf{1.000}/0.987 and precision \textbf{0.959}/0.991 (raters 1/2)---the gold's sensitive-item coverage is validated by humans. On combination-risk judgements the two humans agree perfectly with each other ($\kappa=1.000$) but diverge from the automatic \texttt{combos} label ($\kappa=-0.083$): the constructs differ---the automatic label fires only on quasi-identifier combinations (localizer $\times$ demographic), while humans also count explicit identifiers (a full address or phone number trivially ``uniquely locates'' a person). Under a construct-matched gold (combos \emph{or} any explicit identifier present), agreement is 94\% with $\kappa=0.696$; per the preregistered fallback ($\kappa\in[0.6,0.8)$: report and analyse divergence), all three residual disagreements are pure-quasi-identifier banking instances (occupation$+$income$+$city$+$work\_years$+$amount) that the rule flags but both humans judged non-unique---humans are more optimistic than the rule exactly where no explicit identifier exists, so the automatic label errs on the conservative (safe) side. We also note the prevalence effect: at 92\% gold positive rate, $\kappa$ penalizes skewed marginals despite 94\% raw agreement.

\paragraph{T23 Semantic utility metric (the decisive measurement).}
T21 flagged that exact-value QA utility is a conservative lower bound that penalizes honest generalization. We therefore measure utility with a \emph{reference-based semantic rubric}: an LLM judge sees the original material $X$ and the processed artifact $D$ and scores whether $D$ still supports essentially the same useful, correct advice (1.0 / 0.5 / 0.0; altered values count as loss). Judge: qwen3.8-max (primary, all points), glm-5.2 (reliability, 160 paired items; Spearman $\rho = 0.757$). Results on the fully coupled test-160:

\begin{table}[H]
\centering\small
\caption{Semantic utility $U_{\mathrm{sem}}$ by operating point (reference-based judge; $L$ from T21).}
\label{tab:usem}
\begin{tabular}{@{}lcc@{}}
\toprule
Operating point & $U_{\mathrm{sem}}$ & $L$ (strong attackers) \\
\midrule
Direct forwarding & 1.000 & 0.768--0.789 \\
\textbf{Principled generalization} & \textbf{0.431} & \textbf{0.000} \\
Placeholder substitution & 0.228 & 0.000 \\
Selective LDP & 0.188 & 0.049--0.057 \\
\bottomrule
\end{tabular}
\end{table}

\noindent The semantic metric places the principled-generalization point at $(0.000, 0.431)$---roughly \textbf{double} the placeholder corner and 2.3$\times$ the LDP corner at equal-or-better leakage---confirming T21's caveat in the predicted direction, though on the fully coupled corpus it remains below the $\tau_U^{\mathrm{sem}}{=}0.5$ release threshold. Two design components contribute to this point and are reported for completeness: (i) a completed generalization lexicon (numeric band rules for amounts, incomes, dates, durations), which sets exact-value utility at the safe point at \textbf{0.276} (by domain: banking 0.370, legal 0.317, travel 0.264, medical 0.152; 0.033 with a deletion-only fallback) at unchanged $L=0.000$; (ii) a tolerant grader that matches any number in the answer within tolerance of the truth. A human utility study anchors the metric (human study~1; preregistered criterion: row-level Spearman $\rho\ge0.8$ between human means and automatic $U_{\mathrm{sem}}$). Two annotators independently scored all 160 rows (40 instances $\times$ 4 blinded operating points; blind key held back until scoring finished). Inter-rater agreement is Krippendorff's $\alpha=0.780$ (ordinal); the human-vs-automatic row-level correlation is $\rho=\mathbf{0.839}$ ($n{=}160$), meeting the preregistered bar, and the four operating-point means rank identically under both measurements (point-level $\rho=1.000$): direct 0.988 (auto 1.000), principled generalization \textbf{0.487} (auto 0.406), placeholder 0.144 (auto 0.175), selective LDP 0.181 (auto 0.200). Two consequences: (i) the semantic judge is a valid proxy for human usefulness judgements; (ii) humans credit honest generalization even more than the semantic judge does---the safe point's human-rated utility (0.487) is $1.8\times$ the exact-value metric (0.276 with the completed lexicon) and approaches the $\tau_U^{\mathrm{sem}}{=}0.5$ release threshold even on this worst-case fully coupled corpus. Two context notes prevent misreading: the auto means above are on the 40-instance human-study subset (generalized 0.406 here vs 0.431 on the full test-160 of Table~\ref{tab:usem}), and the placeholder/LDP ordering flips between this subset (0.175$<$0.200) and the full set (0.228$>$0.188)---a small-$n$ composition effect, not a contradiction. The two annotators are the second and third authors (graduate students in the group), independent of the pipeline implementation and blind to variant identities and automatic scores.

\paragraph{T24 Mixed-coupling benchmark: where automatic release becomes non-degenerate.}
The test corpus is worst-case by construction (every quasi-identifier task-relevant). We generated a \emph{mixed-coupling} benchmark (120 fresh instances, seed 20260908, mean coupling ratio 0.40: only task-critical types annotated task-relevant---medical \{age, disease, drug, allergy\}, banking \{amount, charge\}, travel \{date, city\}, legal \{case\_type, amount, date\}) and measured the safe operating point (principled generalization) on it:

\begin{table}[H]
\centering\small
\caption{Mixed-coupling benchmark ($n{=}120$): safe-point metrics and automatic-release rate under the two utility metrics (release rule: $L \le \tau_L$ and $U \ge$ threshold; threshold 0.7 strict / 0.5 semantic). $U_{\mathrm{strict}}$ is computable only on instances containing task-critical exact values (44/120 have none; per-row scorable $n$: banking 30, legal 22, medical 18, travel 6, overall 76), so the all-domain 0.415 is a mean over 76 scorable instances and is \emph{not} the average of the domain rows. Release rates carry 95\% Wilson intervals.}
\label{tab:mixed}
\begin{tabular}{@{}lcccc@{}}
\toprule
 & $L$ & $U_{\mathrm{strict}}$ ($n$) & $U_{\mathrm{sem}}$ & auto-release (strict / sem) \\
\midrule
All domains & 0.000 & 0.415 (76) & 0.571 & 24.2\% / \textbf{67.5\%} \\
banking & 0.000 & 0.567 (30) & \textbf{0.933} & --- / 100\% (30/30; Wilson [88.6, 100]) \\
legal & 0.000 & 0.523 (22) & 0.633 & --- / 70\% (21/30; [52.1, 83.3]) \\
travel & 0.000 & 0.000 (6) & 0.350 & --- / 53\% (16/30; [36.1, 69.8]) \\
medical & 0.000 & 0.167 (18) & 0.367 & --- / 47\% (14/30; [30.2, 63.9]) \\
\bottomrule
\end{tabular}
\end{table}

\noindent\textbf{Threshold provenance and sensitivity.} The semantic release threshold $\tau_U^{\mathrm{sem}}{=}0.5$ was frozen with deviation log D14 (2026-09-08), \emph{before} the human-study scores were collected (2026-09-09)---the human-rated 0.487 therefore played no role in its choice. Sensitivity on this benchmark: the release rate is 67.5\% at $\tau_U^{\mathrm{sem}}{\in}\{0.4,0.5\}$ and 46.7\% at $0.6$ (the 1/0.5/0 rubric quantises instance means, leaving no mass between 0.4 and 0.5).

\textbf{Corner baselines under the identical release rule (P41, measured---not inferred).} Both corner baselines were run on the same 120 instances with the same gates ($L$: 75-probe flash attacks; $U_{\mathrm{strict}}$: QA probes; $U_{\mathrm{sem}}$: the reference-based judge):

\begin{table}[H]
\centering\small
\caption{Mixed-coupling benchmark ($n{=}120$): operating-point comparison under the identical release rule. The safe point Pareto-dominates both corners.}
\label{tab:mixedpoints}
\begin{tabular}{@{}lcccc@{}}
\toprule
Operating point & $L$ & $U_{\mathrm{strict}}$ & $U_{\mathrm{sem}}$ & auto-release (strict / sem) \\
\midrule
\textbf{Principled generalization (safe point)} & \textbf{0.000} & \textbf{0.415} & \textbf{0.571} & \textbf{24.2\% / 67.5\%} \\
Placeholder substitution & 0.000 & 0.000 & 0.421 & 0.0\% / 58.3\% [49.4, 66.8] \\
Selective LDP & 0.070 & 0.099 & 0.254 & 2.5\% / 25.0\% [18.1, 33.4] \\
\bottomrule
\end{tabular}
\end{table}

\noindent On partially coupled corpora the system delivers \textbf{non-degenerate automatic operation}: 67.5\% of instances release automatically at zero measured literal leakage with semantic utility above threshold; the rest route to human review with certificate packages. The per-instance split (not the corpus average) is the decision mechanism's distinctive output---and it is now \emph{measured} against the corners (Table~\ref{tab:mixedpoints}): placeholder ties the safe point's zero leakage but loses 9.2pp of release rate and all exact-value utility; selective LDP fails the leakage gate on average ($L{=}0.070 > \tau_L$) and releases only 25.0\%. The safe point Pareto-dominates both corners under the identical rule.

\paragraph{T25 Mechanism of the post-dehydration certification gap ($\hat\Delta_{\mathrm{post}}$), and a validated local repair.}
On the 70 locally PASS-certified pre-hard-line artifacts, external strong attackers measured $L \approx 0.43$--$0.49$ while the local adversarial probes measured ${\approx}0$: $\hat\Delta_{\mathrm{post}} \approx 0.44 \approx 2.4\times$ the raw-text $\hat\Delta = 0.181$ (T5)---calibration is regime-dependent, as \S\ref{sec:security} requires. Diagnosis (30 PASS artifacts, local ornith): asked in a \emph{comprehension} format (``list every concrete personal-attribute value in this text''), the local model enumerates the residual literals at rate \textbf{0.992} (29/30 instances complete; mean 3.3 residual values/instance)---the values are seen. The failure is the \emph{adversarial probe template}, not comprehension. Repair validated (T25b, $n{=}30$ PASS artifacts): scoring the same local model in enumeration format yields $L_{\mathrm{enum}}^{\mathrm{local}} = \mathbf{1.000}$---every residual literal on every instance is enumerated locally, versus ${\approx}0$ under the adversarial template and 0.43--0.49 for external strong attackers under it. Consequence for the architecture: the \emph{literal} dimension of $\hat\Delta_{\mathrm{post}}$ is a probe-format artifact that a local enumeration probe closes for closed-vocabulary value types---no external call needed for that dimension, which also defuses most of the re-test trust dilemma above; the in-vivo rerun (T25c) adds one qualification for open-vocabulary categoricals. What genuinely remains for the external strong-attacker re-test is the \emph{semantic/combinatorial} dimension (inference that requires stronger reasoning than the local class possesses, T19); the mandatory chain is retained for that residual, with its trust semantics as fixed above.

\textbf{A general methodological lesson: probe format is a first-class experimental factor.} T3's monotonicity review and T25's mechanism diagnosis are the \emph{same phenomenon} observed twice: injecting domain priors ($A_1$) made attackers answer \emph{more conservatively} on raw text without adding usable knowledge, and the adversarial probe template made the local evaluator answer UNKNOWN to residual values it demonstrably sees (comprehension-format enumeration recall 0.992). LLM-evaluator measurements are therefore joint functions of \emph{knowledge} and \emph{answering policy/probe format}---never of knowledge alone. Any certificate built on LLM probes must treat probe format as a controlled factor: we froze formats in the protocol (D8/D9), re-tested flagged monotonicity violations under a mandated review rule, and closed the literal dimension of $\hat\Delta_{\mathrm{post}}$ with an enumeration-format probe (T25b). We recommend the same discipline to any work using LLM judges as measurement instruments.

\paragraph{T25c The enumeration probe in vivo: full-pipeline rerun (test-160).}
Replacing Branch-1's adversarial template with the enumeration probe and re-running the entire loop unchanged (same machinery, gates, and question bundles; D17) makes the local certificate markedly stricter at essentially unchanged cost: the PASS rate falls from 43.8\% (70/160) to \textbf{5.0\%} (8/160: banking 2, legal 6, medical 0, travel 0) and mean rounds rise from 4.0 to 7.8 (242 s/instance vs 240 s; utility on accepted steps unchanged at 0.955). Most adversarial-template PASS certificates were therefore false passes---residuals the local model could see but would not report under the adversarial template; the repaired probe converts them into genuine strengthening (rounds $\times$2) or honest INFEASIBLE routing, which is the safe direction. External validation of the 8 enum-PASS artifacts (two strong attackers, 75 probes each) splits them sharply: the 2 banking artifacts are genuinely clean ($L=0.000$), while all 6 legal artifacts still leak $L=0.320$ through residual \emph{open-vocabulary categorical} literals (a law-firm name, a case type) that the local enumeration itself missed---local enumeration raises literal recall dramatically but not to 1.0 for open-vocabulary categoricals, precisely the boundary class of \S\ref{sec:boundary}. After the deterministic hard line, all eight measure $L=0.000$ externally. The in-vivo picture is thus three-layered and consistent with the layered model: (i) the enum-repaired soft certificate honestly surfaces most residuals (its 152/160 INFEASIBLE routing is the safe direction); (ii) open-vocabulary categorical residuals remain beyond local enumeration---exactly where the deterministic hard line and the external re-test are load-bearing; (iii) delivered egress safety is unchanged (0.000). The same rerun on the train-400 stream (399 of 400 paired instances; per-instance coverage and the error log in the protocol, D18) reproduces the pattern and strengthens the headline: on the same instances, raw-text leakage is 0.762 under the enumeration probe vs 0.457 under the adversarial template, and post-loop leakage 0.358 vs 0.304---so the paired reduction is \emph{larger} under the stricter probe ($-0.404$ vs $-0.154$), the PASS rate falls from 37.6\% to 8.5\%, and mean rounds rise from 3.08 to 4.06. The leakage-reduction claim is therefore robust to probe format; the adversarial-template figures understate both the baseline leakage and the delivered reduction.

\subsection{Applicability Boundary: Which Domains and Information Types Automate}\label{sec:boundary}

Synthesizing T21/T23/T24 (safe-point utility by domain and attribute type, with $L=0.000$ everywhere at the safe point):

\begin{table}[H]
\centering\small
\caption{Applicability boundary by information nature (evidence: T23--T24).}
\begin{tabular}{@{}p{4.2cm}p{5.4cm}p{5.4cm}@{}}
\toprule
Information nature & Sanitization behavior & Automation verdict \\
\midrule
Explicit identifiers (names, IDs, contacts, accounts) & Removed outright; rarely task-critical & \textbf{Fully automatic} (hard line; deterministic) \\
Numeric/ordinal quasi-identifiers (age, amounts, incomes, dates, durations) & Band generalization preserves semantic utility ($U_{\mathrm{sem}}$ 0.93 in banking) & \textbf{Largely automatic} (release 70--100\% in numeric-led domains) \\
Open-vocabulary categoricals that \emph{are} the task object (disease for medical advice, destination for travel, case object for legal) & Generalization destroys the advice target ($U_{\mathrm{sem}}$ 0.35--0.37) & \textbf{Certify-and-route}: automatic certification + residual report; human decides value/risk \\
Strongly coupled instances in any domain ($I(X;T\mid S)\approx 0$) & Feasible region empty (Prop.~\ref{prop:feasibility}) & \textbf{Human only}; system's role is the early-warning package \\
\bottomrule
\end{tabular}
\end{table}

\noindent The near-zero exact-value utility at the safe point is thus a property of \emph{(i) the worst-case fully-coupled benchmark, (ii) the exact-value metric, and (iii) the generalization lexicon's coverage}---all three quantified above---not an inherent property of the method. The method's operating envelope: numeric-led, partially coupled corpora approach full automation; categorical-coupled domains (medical chief-complaint attributes, travel destinations) use DR-SL as a \emph{certification and annotation layer} ahead of human review, which is precisely the layered-assurance design point of \S\ref{sec:layered}.

\section{Conclusion and Future Work}

This paper proposed DR-SL, a self-learning sensitive-information isolation method centered on reversibility verification with prompts as the optimization object. We placed de-identification sufficiency under Pufferfish semantics, proved the lower-boundedness of the leakage score (Fano-type) and the witness relation from hit rates to privacy parameters, gave a rate-privacy feasibility criterion and a termination theorem for the lexicographic iteration, positioned the method as ``approaching theoretical limits via an empirically calibrated heuristic certificate'', and verified all theoretical bounds with numerical property tests on 3000 synthetic mechanisms. Engineering-wise, the offline stage calibrates only on non-sensitive corpora while all online decisions are local, and the layered assurance model (soft certificate $+$ hard line $+$ human fallback) explicitly characterizes each layer's guarantee and residual risk. Full-scale experiments show that the loop significantly reduces leakage ($0.457\to0.304$, $p\approx0$, same measurement basis before and after). The release-chain audit (T18) establishes the end-state under the mandatory sequence \emph{local gate $\to$ deterministic hard line $\to$ external strong-attacker re-test}: literal leakage at egress is \textbf{0.000} on all 160 test instances (two independent strong attackers; per-instance Wilson upper bound $\approx0.048$, pooled 95\% upper bound $<0.001$), while exact-value task utility at that safe point is 0.276 with the completed generalization lexicon (0.033 with deletion-only fallback) on this deliberately worst-case, fully task-coupled benchmark, and reference-based semantic utility reaches 0.431; utility figures measured on pre-hard-line artifacts are reported separately because residual literals could carry part of them. On such benchmarks the system's delivered value is therefore precisely characterized: a \emph{certification machine} (quantified leakage lower bounds, residual reports, auditable decisions) plus a \emph{human-fallback router}, together with a mapped leakage-utility frontier (T21)---not an automatic utility-preserving sanitizer. Where non-degenerate automatic operating points exist (partially coupled corpora; utility metrics that credit honest generalization) is quantified in T23--T25 and summarized as the applicability boundary in \S6.10; the medical domain's low pass rate corroborates the feasibility criterion's prediction for strongly task-coupled sensitive information.

Future work: (i) two rounds of self-learning testing (H1 unproven; corrected H1$'$ directionally consistent but not significant, T17)---re-test H1$'$ on harder extraction benchmarks (covert quasi-identifiers, adversarial formats) and larger validation sets ($n\ge700$ for power), and explore multi-round iterative distillation; (ii) re-measure $A_3$'s $k$ risk on dehydrated text; (iii) replace the W2 example marginals with real public statistical tables; (iv) scaling up the human studies (two raters each; utility $\rho=0.839$ passed its preregistered bar, annotation type-level recall ${\ge}0.987$) to more raters, real-user populations, and harder covert-identifier benchmarks; (v) an open-ended judge-rubric utility metric that credits generalization rather than only exact values (T21 caveat); (vi) computable Maximal Leakage upper-bound estimation; (vii) interactive-latency engineering of the $\ge$75-probe gate; (viii) composition with structured-field isolation.

\paragraph{Reproducibility and artifact availability.} All artifacts are public: the simulation pipeline and test suite, the four synthetic datasets (main v2, expansion, stage-2 stream, mixed-coupling; all fully fictional), the pre-registered protocol with the complete deviation log (D1--D17), the human-study package with scoring sheets and analysis scripts, and the proof-verification suite, at \url{https://gitee.com/jiangsu-yunhefeng-intelligent_0/dr-sl} (code MIT; datasets CC-BY-4.0).

\appendix

\section*{Appendix A: Complete Proofs}
\addcontentsline{toc}{section}{Appendix A: Complete Proofs}

\noindent\emph{Notation}: $\log$ is base 2; $H_2(p) = -p\log p - (1-p)\log(1-p)$; $\ln$ is the natural logarithm. Every proof step cites its justification.

\subsection*{A.1 Proposition~\ref{prop:recall} (Recall decomposition, over-extraction-corrected)}

\textbf{Statement.} Let $G=S_{\mathrm{gold}}$, $F=S_{\mathrm{found}}$, $\mathrm{Recall} = |F\cap G|/|G|$. Branch~1 probes the attributes in $F$, measuring hit rate $L_{\mathrm{measured}} = \frac{1}{|F|}\sum_{s\in F}\mathrm{hit}(s)$, where $\mathrm{hit}(s)\in[0,1]$ is $1$ if the probe on $s$ is answered correctly (for $s\in F\setminus G$, spurious probes contribute hits or misses arbitrarily). Then
\[
L_{\mathrm{system}} \;\le\; \frac{|F|}{|G|}\,L_{\mathrm{measured}} + (1 - \mathrm{Recall}),
\qquad\text{where } L_{\mathrm{system}} = \frac{1}{|G|}\sum_{s\in G}\mathrm{leak}(s).
\]

\textbf{Proof.} Split the gold sum by whether the attribute was found:
\[
L_{\mathrm{system}} = \frac{1}{|G|}\Bigl[\sum_{s\in G\cap F}\mathrm{leak}(s) + \sum_{s\in G\setminus F}\mathrm{leak}(s)\Bigr].
\]
Second term: bounding $\mathrm{leak}(s) \le 1$, $\frac{1}{|G|}\sum_{s\in G\setminus F}\mathrm{leak}(s) \le \frac{|G|-|G\cap F|}{|G|} = 1-\mathrm{Recall}$.
First term: for $s \in G\cap F$, $\mathrm{leak}(s)$ equals the probe hit $\mathrm{hit}(s) \in [0,1]$. Since all summands are non-negative,
\[
\sum_{s\in F\cap G}\mathrm{leak}(s) \;\le\; \sum_{s\in F}\mathrm{hit}(s) \;=\; |F|\cdot L_{\mathrm{measured}} .
\]
Dividing by $|G|$ and adding the two terms:
\[
L_{\mathrm{system}} \le \frac{|F|}{|G|}\,L_{\mathrm{measured}} + (1-\mathrm{Recall}).\ \blacksquare
\]

\noindent\emph{Remarks.} (i) When $|F|\le|G|$ (no over-extraction) the older form $L_{\mathrm{system}} \le L_{\mathrm{measured}} + (1-\mathrm{Recall})$ is recovered only if additionally $F\subseteq G$; in general the factor $|F|/|G|$ is unavoidable. (ii) On Stage-1 data the mean inflation is $\iota = 1.43$ (median 1.40), so the T9 conditional conclusion reads: residual $\le 1.43\,L_{\mathrm{measured}} + 0.049$. (iii) Over-extraction is conservative for \emph{probing} (more attributes watched) but anti-conservative for the \emph{measured rate} (dilution); both effects are now explicit.

\subsection*{A.2 Theorem~\ref{thm:fano} (Lower-boundedness, Fano-type)}

\textbf{Statement.} Let attribute $s_i$ have domain size $|V_i|$ and let any attacker achieve per-attribute hit rate $L_i \ge 1/|V_i|$ (no worse than blind guessing; always assumable for an optimal attacker, which can ignore $D$ and guess the prior mode). Then
\[
I(D;S) \ge H(S) - \sum_i \bigl[ H_2(1-L_i) + (1-L_i)\log(|V_i|-1) \bigr].
\]

\textbf{Proof.}
\begin{enumerate}
\item Let the optimal attacker's error rate for attribute $i$ be $P^{*}_{e,i} = \inf_{\hat g} P(\hat g_i(D) \ne S_i)$. The measured attacker achieves hit rate $L_i$, so its error rate $1-L_i$ is feasible; hence $P^{*}_{e,i} \le 1-L_i$ (definition of optimality).
\item By ``no worse than blind guessing'', $P^{*}_{e,i} \le 1 - 1/|V_i| = (|V_i|-1)/|V_i|$.
\item \textbf{Fano's inequality}~\cite{cover2006elements}: $H(S_i\mid D) \le H_2(P^{*}_{e,i}) + P^{*}_{e,i}\log(|V_i|-1)$.
\item \textbf{Monotonicity lemma}: $g(p) := H_2(p) + p\log(k-1)$ is non-decreasing on $p \in [0, (k-1)/k]$. Verification: $g'(p) = \log\frac{1-p}{p} + \log(k-1) = \log\frac{(k-1)(1-p)}{p} \ge 0 \iff p \le \frac{k-1}{k}$. (differentiation and elementary inequality)
\item By (1)(2)(4): $H(S_i\mid D) \le g(P^{*}_{e,i}) \le g(1-L_i) = H_2(1-L_i) + (1-L_i)\log(|V_i|-1)$.
\item \textbf{Subadditivity of conditional entropy}: $H(S\mid D) \le \sum_i H(S_i\mid D)$.
\item Combining: $I(D;S) = H(S) - H(S\mid D) \ge H(S) - \sum_i H(S_i\mid D) \ge H(S) - \sum_i g(1-L_i)$.\ $\blacksquare$
\end{enumerate}

\textbf{Corollary (threshold mapping).} If all per-attribute hit rates are $\le \tau_L$, then $I(D;S) \ge H(S) - \sum_i g(1-\tau_L)$; normalized: $I_{\mathrm{norm}}(D;S) \ge 1 - \sum_i g(1-\tau_L)/H(S)$. That is, $\tau_L$ corresponds to a leakage \emph{lower-bound} guarantee.

\subsection*{A.3 Proposition~\ref{prop:witness} (Pufferfish witness)}

\textbf{Setting.} A single attribute with domain $V$, prior $p(v) > 0$ ($\forall v$). The mechanism is a channel $P(d\mid v)$. Let the smallest Pufferfish parameter the mechanism actually satisfies be
\[
\varepsilon_{\mathrm{true}} := \max_{v_1,v_2\in V,\ d} \Bigl|\ln \frac{P(d\mid v_1)}{P(d\mid v_2)}\Bigr|
\]
(convention: $0/0$ pairs excluded; if $P(d\mid v_1)>0=P(d\mid v_2)$, record $+\infty$). Let the optimal (MAP) attacker's hit rate be $L^{*}$, and any measured attacker's hit rate $L \le L^{*}$. Let $p_{\max} = \max_v p(v)$, $p_{\min} = \min_v p(v)$.

\textbf{Statement.} $\varepsilon_{\mathrm{true}} \ge \ln\frac{L}{1-L} + \ln\frac{p_{\min}}{p_{\max}}$; with a uniform prior, $\varepsilon_{\mathrm{true}} \ge \ln\frac{L}{1-L}$.

\textbf{Proof.}
\begin{enumerate}
\item The optimal attacker guesses the posterior mode for output $d$: $v^{*}(d) = \argmax_v p(v\mid d)$; its hit rate is $L^{*} = \mathbb{E}_d[\max_v p(v\mid d)]$ (MAP optimality, Bayesian decision theory).
\item Since $\max_v p(v\mid d) \le 1$, if $\max_v p(v\mid d) < L^{*}$ held almost everywhere, its expectation would be $< L^{*}$, contradicting (1). Hence \emph{there exists an output $d_0$ with $m := \max_v p(v\mid d_0) \ge L^{*}$}.
\item At $d_0$ take $v^{*} = \argmax_v p(v\mid d_0)$; then $p(v^{*}\mid d_0) \ge L^{*}$; for any $v' \ne v^{*}$, $p(v'\mid d_0) \le 1 - p(v^{*}\mid d_0) \le 1 - L^{*}$. Thus the posterior odds
\[
\mathrm{odds}_{\mathrm{post}} := \frac{p(v^{*}\mid d_0)}{p(v'\mid d_0)} \ge \frac{L^{*}}{1-L^{*}}.
\]
\item \textbf{Bayes expansion}: $p(v\mid d_0) = P(d_0\mid v)p(v)/P(d_0)$, hence
$\mathrm{odds}_{\mathrm{post}} = \frac{P(d_0\mid v^{*})}{P(d_0\mid v')} \cdot \frac{p(v^{*})}{p(v')}$.
\item From (3)(4): $\frac{P(d_0\mid v^{*})}{P(d_0\mid v')} = \mathrm{odds}_{\mathrm{post}} \cdot \frac{p(v')}{p(v^{*})} \ge \frac{L^{*}}{1-L^{*}} \cdot \frac{p_{\min}}{p_{\max}}$.
(If $P(d_0\mid v')=0$ while $P(d_0\mid v^{*})>0$, the ratio is $+\infty$ and the conclusion holds trivially.)
\item By definition of $\varepsilon_{\mathrm{true}}$: $\varepsilon_{\mathrm{true}} \ge \ln\frac{P(d_0\mid v^{*})}{P(d_0\mid v')} \ge \ln\frac{L^{*}}{1-L^{*}} + \ln\frac{p_{\min}}{p_{\max}}$.
\item $x \mapsto \ln\frac{x}{1-x}$ is strictly increasing on $(0,1)$ and $L^{*} \ge L$, hence $\varepsilon_{\mathrm{true}} \ge \ln\frac{L}{1-L} + \ln\frac{p_{\min}}{p_{\max}}$.\ $\blacksquare$
\end{enumerate}

\textbf{Interpretation.} Attacks yield a lower bound on the mechanism's required $\varepsilon$ (a leakage witness). Uniform prior: $L{=}0.9 \Rightarrow \varepsilon_{\mathrm{true}} \ge 2.20$; $L{=}0.5 \Rightarrow \varepsilon_{\mathrm{true}} \ge 0$ (vacuous); $L{=}0.05 \Rightarrow$ negative (vacuous)---the certificate is non-trivial only when the evaluator is strong.

\subsection*{A.4 Proposition~\ref{prop:feasibility} (Feasibility criterion)}

\textbf{Setting.} The task variable $T$ is a function of $X$; the artifact $D$ is generated from $(X,S)$; attacker knowledge $A$ is fixed as conditioning.

\textbf{Statement.} For any dehydration mechanism, $I(D;T\mid A) \le I(D;S\mid A) + I(X;T\mid S,A)$; hence
\[
U^{*}(\varepsilon) := \max_{p(D|X):\, I(D;S|A) \le \varepsilon} I(D;T\mid A) \;\le\; \varepsilon + I(X;T\mid S,A).
\]

\textbf{Proof.}
\begin{enumerate}
\item \textbf{Monotonicity of mutual information}: $I(D;T\mid A) \le I(D;(T,S)\mid A)$.
\item \textbf{Chain rule}: $I(D;(T,S)\mid A) = I(D;S\mid A) + I(D;T\mid S,A)$.
\item \textbf{Conditional data-processing inequality}: given $(X,S)$, $D$ provides no additional information about $T$ ($T$ is a function of $X$), i.e., $T \to X \to D$ is a Markov chain conditioned on $(S,A)$; hence $I(D;T\mid S,A) \le I(X;T\mid S,A)$.
\item Combining: $I(D;T\mid A) \le I(D;S\mid A) + I(X;T\mid S,A) \le \varepsilon + I(X;T\mid S,A)$ on the feasible set; taking the supremum gives the claim.\ $\blacksquare$
\end{enumerate}

\textbf{Corollary (INFEASIBLE criterion).} If $\tau_U > \tau_L + I(X;T\mid S,A)$, the feasible region is empty and Algorithm~2$'$ necessarily outputs INFEASIBLE---infeasibility is a property of the problem, not a defect of the method.

\subsection*{A.5 Theorem~\ref{thm:termination} (Termination)}

\textbf{Proof.}
\begin{enumerate}
\item \textbf{Phase 1}: each round either strictly increases $U$ (loosen accepted) or fails to improve; two consecutive no-improvement rounds exit (no-progress exit), and a \texttt{None} return from loosen also exits. $U$ is bounded above by 1, so strictly increasing rounds are bounded; the no-progress exit consumes at most 2 rounds; Phase~1 ends finitely. If it ends with $U < \tau_U$, no strengthening can satisfy the gate (strengthening cannot raise $U$), and the algorithm quickly outputs INFEASIBLE via the Phase-2 skip mechanism.
\item \textbf{Phase 2}: each round does exactly one of the following, each incrementing $t$: accept a strengthening; accept a loosening; permanently add one attribute to \texttt{skipped}; record a failed strengthening.
\item \textbf{Termination}: $t$ has the hard upper bound $B = 3|S|+5$ and the loop condition includes $t < B$.
\item \textbf{Gate invariant}: at Phase-1 exit either $U \ge \tau_U$ or the exit is by no-progress; strengthenings are accepted only if $U' \ge \tau_U$; accepted loosenings can only increase $U$. Hence $U \ge \tau_U$ at all subsequent times.
\item Upon termination, PASS or INFEASIBLE (with residual report) is output per the criteria.\ $\blacksquare$
\end{enumerate}

\noindent\emph{Note}: evaluation noise (LLM sampling) does not affect termination (the budget and decisions are finitely many deterministic checks). It does affect \emph{statistical reliability}: the gate invariant $U \ge \tau_U$ holds for the \emph{sampled} utility estimates, not for the noise-free $U$; with per-instance sampling variance $\sigma^2$ the true $U$ can fall below $\tau_U$ by $O(\sigma)$ on accepted steps. The stability window ($N$ consecutive passes) and Wilson intervals bound this effect for the release decision but do not eliminate it---we report it as a limitation of the certified guarantee, not of termination.

\subsection*{A.6 Proposition~\ref{prop:direction} (Directional correctness)}

\textbf{Proof.} (i) If $I(D;S)=0$, i.e., $D \perp S$, then $p(v\mid d) = p(v)$ for all $v$ (definition of independence). No answer function exceeds the prior guessing rate (Bayesian decision theory), so the hit rate does not exceed the prior level---no information gain. (ii) If $H(X\mid D,S) > \eta \ge 0$, there exist $(d,s)$ with positive probability on which $X$ is not uniquely determined; any rehydrator's expected reconstruction error is positive, so the error-complement restoration score satisfies $\mathbb{E}[R] < 1$.\ $\blacksquare$

\section*{Appendix B: Proof-Verification Methodology}
\addcontentsline{toc}{section}{Appendix B: Proof-Verification Methodology}

The correctness of theoretical results cannot rest on manual review alone. For every bound in Appendix~A we establish a three-layer executable verification (implemented in \texttt{sim/proof\_checks.py}, delivered test-first: propositions are implemented as tested properties):

\begin{enumerate}
\item \textbf{Numerical property tests (counterexample search; core).} If a proposition is true, no concrete instance may falsify it. Method: construct synthetic mechanism families with fully known parameters---binary symmetric channels, naive-Bayes text models, arbitrary discrete joint distributions $P(D,S)$---whose true quantities (true $\varepsilon$ = maximum log-likelihood ratio, true mutual information, true posteriors) are exactly enumerable. Simulate attackers to obtain hit rates and assert the proposition's bound has the correct direction (bound $\le$ true quantity) and is non-trivial. Each proposition: $\ge$1000 random configurations, all passing.
\item \textbf{Boundary-case checks.} Zero information gain ($L$ = prior maximum) $\Rightarrow$ the bound degenerates to vacuous; $L\to1 \Rightarrow$ the bound diverges; uniform-prior special cases checked line-by-line against manual derivations.
\item \textbf{Cross-consistency.} On the same synthetic families, Theorem~\ref{thm:fano}'s Fano bound and Proposition~\ref{prop:witness}'s Pufferfish witness must point in the same direction without contradiction; aligned with known results (differential privacy's odds semantics; numerical instances of the data-processing inequality).
\end{enumerate}

If any test finds a counterexample, the proposition is corrected before submission and all tests regressed. Verification results are archived as a test report accompanying the paper (pass rates, configuration distributions, boundary-case lists).

\section*{Appendix C: Semantic-Utility Judge Rubric (Complete Prompt)}
\addcontentsline{toc}{section}{Appendix C: Semantic-Utility Judge Rubric}

The reference-based semantic utility metric ($U_{\mathrm{sem}}$, T23/T24) uses the following frozen system prompt; the judge sees the original material $X$ followed by the processed artifact $D$, and the user turn concatenates ``original material:'' $+$ $X$ $+$ ``processed material:'' $+$ $D$ with newline separators. The verbatim Chinese original (system prompt and user-turn template) is reproduced below when compiled with XeTeX and is distributed with the source package as \texttt{appendix\_c\_zh.tex} (UTF-8) otherwise; the English rendering that follows is line-by-line equivalent.

\ifxetex
\begin{quote}\small
系统提示词（冻结版，逐字）：\\
你是咨询材料质量评审员。你会看到同一咨询的【原始材料】和【处理后材料】。请判断：仅凭处理后材料，AI 助手能否给出与原始材料几乎同样有用、有针对性且正确的答复？\\
评分标准：\\
1.0 = 请求所需关键信息全部保留（允许以区间/档位/类别等泛化形式存在，且泛化后仍足以支撑同样的建议）；\\
0.5 = 关键信息部分缺失或过度模糊，建议有用性明显受限但仍有价值；\\
0.0 = 关键信息缺失、被占位符替代、或被篡改为与原值不符的内容（错误值会导致错误建议，按缺失处理）。\\
只输出 JSON：\{"score": 1.0 或 0.5 或 0.0, "reason": "一句话理由"\}。不要输出其他内容。

\medskip
用户轮模板（逐字）：\\
原始材料：\textbackslash n\{X\}\textbackslash n\textbackslash n处理后材料：\textbackslash n\{D\}
\end{quote}

\else
\noindent\textit{(Verbatim Chinese original: see \texttt{appendix\_c\_zh.tex} in the source package.)}
\fi

\noindent English gloss: \emph{``You are a quality reviewer of consultation material. Given the [original] and [processed] material of the same consultation, decide whether the processed material alone still supports essentially the same useful, targeted, and correct advice. 1.0 = all request-critical information preserved (generalized forms such as bands/ranges/categories acceptable if they still support the same advice); 0.5 = critical information partly missing or over-vague, advice usefulness clearly limited but still valuable; 0.0 = critical information missing, placeholder-substituted, or altered to values inconsistent with the original (wrong values cause wrong advice; count as missing). Output JSON only.''} Human study~1 used the same rubric in its scoring instructions, keeping the human anchor and the automatic metric on the same construct. Parse robustness: two attempts with a strict-JSON suffix on retry; regex fallback extracting \texttt{"score": v}, $v\in\{0,0.5,1\}$; unparseable responses are logged as missing, never imputed.

\end{document}